\documentclass[11pt,a4paper,twoside,reqno]{amsart}
\renewcommand\thesection{\arabic{section}}

\makeatletter
\def\@settitle{\begin{center}%
		\baselineskip14\p@\relax
		\normalfont\LARGE\scshape\bfseries
		\@title
	\end{center}%
}

\def\@setauthors{%
  \begingroup
  \def\thanks{\protect\thanks@warning}%
  \trivlist
  \centering\footnotesize \@topsep30\p@\relax
  \advance\@topsep by -\baselineskip
  \item\relax
  \author@andify\authors
  \def\\{\protect\linebreak}%
  \authors%
  \ifx\@empty\contribs
  \else
    ,\penalty-3 \space \@setcontribs
    \@closetoccontribs
  \fi
  \endtrivlist
  \endgroup
}

\makeatother

\makeatletter

\def\subsection{\@startsection{subsection}{2}%
	\z@{.5\linespacing\@plus.7\linespacing}{.5\linespacing}%
	{\normalfont\large\bfseries}}

\makeatother

\makeatother

\usepackage[usenames, dvipsnames]{color}
\definecolor{darkblue}{rgb}{0.0, 0.0, 0.45}

\usepackage[colorlinks	= true,
raiselinks	= true,
linkcolor	= darkblue, 
citecolor	= Mahogany,
urlcolor	= ForestGreen,
pdfauthor	= {Peyman Mohajerin Esfahani},
pdftitle	= {},
pdfkeywords	= {},
pdfsubject	= {},
plainpages	= false]{hyperref}

\usepackage{dsfont,amssymb,amsmath,subfigure, graphicx,enumitem,multirow} 
\usepackage{amsfonts,dsfont,mathtools, mathrsfs,amsthm} 
\usepackage[amssymb, thickqspace]{SIunits}
\usepackage{algorithm,fancyhdr}
\usepackage{algorithmicx}
\usepackage{algpseudocode}

\allowdisplaybreaks
\date{\today}
\renewcommand\thesection{\arabic{section}} 

\theoremstyle{theorem}
\newtheorem{theorem}{Theorem}[section]
\newtheorem{proposition}[theorem]{Proposition}

\newtheorem{lemma}[theorem]{Lemma}

\newtheorem{assumption}[theorem]{Assumption}
\newtheorem{problem}[theorem]{Problem}
\newtheorem{remark}[theorem]{Remark}

\newtheorem{definition}[theorem]{Definition}

\newcommand{\R}{\mathbb{R}}

\newcommand{\N}{\mathbb{N}}

\newcommand{\diff}{\mathrm{d}}

\usepackage{epsfig} 
\usepackage{graphicx}
\usepackage{grffile}
\usepackage{amsmath}
\usepackage{enumitem}
\usepackage{booktabs}
\usepackage{multirow}
\usepackage{array}
\usepackage{cite}
\usepackage{bm}
\usepackage{pdfpages} 
\usepackage{subfigure}
\usepackage{booktabs}
\usepackage{xcolor}
\usepackage{algorithm}
\usepackage{algpseudocode}
\usepackage{wrapfig}
\usepackage{caption}

\usepackage[colorinlistoftodos]{todonotes}

\begin{document}

\title[BIA Detection]{Stealthy bias injection attack detection based on Kullback-Leibler divergence in noisy linear systems}
\author{Jingwei Dong, André M. H. Teixeira}
\thanks{Jingwei Dong and André M. H. Teixeira are with the Department of Information Technology, Uppsala University, SE-75105 Uppsala, Sweden. This work is supported by the Swedish Research Council under
the grants 2021-06316 and 2023-05234, the Swedish
Foundation for Strategic Research, and the Knut and Alice Wallenberg
Foundation.}

\begin{abstract}
    This paper studies the design of observer gains to enhance the detectability of stealthy bias injection attacks (BIAs) in linear systems driven by Gaussian noise, where adversaries compromise a subset of sensors and exploit system dynamics and noise to remain stealthy.
We formulate the observer-gain design as a max-min optimization problem that enhances the detectability of the least detectable BIAs, with detectability measured by the Kullback-Leibler divergence (KLD) subject to a prescribed attack impact. 
To reduce the computational complexity of the derived non-convex design problem, we consider attack onset, one-step, and \mbox{steady-state} detectability. 
We show that the Kalman filter is optimal at attack onset regardless of the subset of attacked sensors. 
For the one-step and steady-state cases, the observer gain design problems are approximated by bi-convex optimization problems, which can be efficiently solved using alternating optimization (AO) and alternating direction method of multipliers (ADMM).
Moreover, more tractable linear matrix inequality (LMI) relaxations are developed.
Finally, the effectiveness of the proposed stealth-aware detection framework is demonstrated through an application to a thermal system.
\end{abstract}

\maketitle

\section{Introduction}
Cyber-physical systems (CPSs) enable high-performance operation in critical infrastructures by tightly integrating physical processes with
communication, computation, and control.
However, their reliance on networked sensing and communication exposes them to cyber threats that can compromise system performance.
Conventional information-theoretic security mechanisms provide essential protection but ignore the physical dynamics, and may therefore be insufficient against cyber attacks that exploit system knowledge to remain stealthy while causing physical damage~\cite{jiang2023monitoring}.
Consequently, increasing attention has been devoted to the cybersecurity of CPSs from the control-theoretic perspective~\cite{mo2009secure,teixeira2015secure}.

To secure CPSs against cyber attacks, model-based detection and resilient
estimation have been extensively studied~\cite{liu2023secure,lu2024attack}.
In this work, we focus on the detection of bias injection attacks (BIAs), a class of false data injection (FDI) attacks in which adversaries compromise a subset of sensors and inject bias signals to mislead controllers or state estimators. 
Although structurally simple, BIAs can gradually degrade estimation or control performance and have received considerable attention in applications such as power systems~\cite{
liu2024detection}.
Moreover, they can be strategically crafted by exploiting system dynamics and noise, and injected into sensors in a coordinated manner to achieve malicious objectives while remaining stealthy~\cite{teixeira2015secure}, making their detection under noisy conditions challenging.

\subsection{Related work}
Attack detection methods in CPSs can be broadly classified into active and passive approaches.
Specifically, active detection methods, such as moving target defense and watermarking, aim to proactively perturb system dynamics or inputs to expose malicious behaviors like zero-dynamics attacks and replay attacks~\cite{liu2023proactive,naha2023quickest}.
Passive attack detection methods, mostly derived from residual \mbox{generation-based} fault diagnosis frameworks, including Kalman filter-based statistical detectors~\cite{manandhar2014detection,murguia2019model} and observer-based attack estimation~\cite{huang2022observer,ahmed2024detection}, are commonly used to detect FDI attacks.
Under the assumption of bounded measurement noise, zonotopic \mbox{set-membership} methods have also been adopted to detect bias attacks by checking the consistency between predicted and measured state sets~\cite{li2023attack}.
However, the aforementioned methods generally treat FDI attacks as ordinary anomalies and do not explicitly incorporate attack stealthiness and impact into detector design, which limits their effectiveness against deliberately stealthy bias-type attacks.

This motivates the explicit characterization of attack stealthiness in detector design. 
Attack stealthiness can be viewed from both structural and stochastic perspectives, with undetectability representing a stronger notion in which an attack produces no distinguishable detection signature.
In noiseless systems, undetectable attacks can be constructed by exciting only the system zero dynamics~\cite{pasqualetti2013attack}.
In stochastic systems, stealthiness can be characterized by the statistical similarity between attacked and attack-free measurements. 
Early studies considered stealthiness relative to specific detectors, such as bad-data detection and the $\chi^2$ test~\cite{cui2012coordinated,mo2013detecting}.
The Kullback-Leibler divergence (KLD) was introduced in \cite{bai2017data} as a general measure of statistical distinguishability between attacked and attack-free systems. 
Moreover, the KLD admits a closed-form expression under Gaussian assumptions and has been widely used in the analysis and design of stealthy FDI attacks in CPSs~\cite{guo2018worst,zhang2019optimal}.

Compared with attack strategy design,  stealthiness has received less attention in detector synthesis.
Existing studies mainly address stealthy attacks in power systems using statistical methods~\cite{chaojun2015detecting,deng2015defending}, or structurally undetectable attacks such as zero-dynamics attacks that typically require active detection mechanisms~\cite{teixeira2012revealing}.
Relatively few studies have explicitly incorporated stealthiness into residual generation-based detector design.
Our previous work~\cite{tosun2025kullbackEJC} used KLD to characterize worst-case attack detectability and design observers against stealthy BIAs. 
\cite{cheng2021security} considered stealthy FDI attack detection in power systems under stochastic measurement delays using a modified Kalman filter and a delay-resilient detector. 
More recently, \cite{feng2025false} embedded stealthiness constraints into detector design using Wasserstein ambiguity sets.
Nevertheless, \mbox{stealth-aware} residual-generation-based detector design is still underexplored, which motivates the framework developed in this work.

\subsection{Main contributions}
Based on the above analysis, this paper investigates the detection of BIAs targeting a subset of sensors in linear \mbox{time-invariant} (LTI) systems driven by Gaussian noise. 
We propose an observer-based detection framework in which both attack stealthiness and attack impact are explicitly incorporated into detector design.
The contributions of this paper are summarized as follows: 
\begin{itemize}
    \item {\bf Worst-case BIA detector design:} 
    We formulate the observer-based BIA detector design as a max-min problem that maximizes the KLD-based detectability of the \mbox{worst-case} BIAs, defined as attack vectors with the lowest detectability subject to a prescribed impact constraint.

    \item {\bf Multi-instant KLD-based design:} 
    To obtain tractable synthesis problems, we consider KLD-based detectability measures at several detection instants, including attack onset, one step after attack occurrence, and the steady-state scenario. 
    Notably, we show that the Kalman filter is the optimal solution for the attack onset case, regardless of the compromised sensor set and the impact weighting matrix (Proposition~\ref{prop: kld0}). 

    \item {\bf Observer-gain design via bi-convex approximation:}
    For the one-step and steady-state cases, we derive \mbox{bi-convex} approximations of the original non-convex observer-gain design problems (Theorems~\ref{thm: KLD1_BMI} and~\ref{thm: KLDinf BMI}). 
    Alternating optimization (AO) and the alternating direction method of multipliers (ADMM) are provided to solve the resulting problems.

    \item{\bf Tractable linear matrix inequality (LMI) approximation:} 
    The observer-gain design problems for one-step and steady-state cases are further approximated by LMI conditions (Theorems~\ref{thm: KLD1LMI} and~\ref{thm: KLDinf_LMI}). 
    This leads to computationally tractable synthesis procedures at the cost of increased conservatism.
\end{itemize}

This work constitutes an extension of our previous research~\cite{tosun2025kullbackEJC}. The main differences are summarized as follows:
(1) In~\cite{tosun2025kullbackEJC}, all sensors are assumed to be potentially compromised, whereas this paper introduces a structured matrix to characterize partially corrupted scenarios, which better reflects the practical situation, as some secured sensors are harder to access;
(2) The attack impact in~\cite{tosun2025kullbackEJC} was characterized using a positive definite weighting matrix. 
In contrast, this paper relaxes this assumption and allows the weighting matrix to be positive semidefinite, rendering the methods in~\cite{tosun2025kullbackEJC} inapplicable under this more general setting;
(3) The previous work focuses on detector design at the attack onset and the steady-state regime. 
This paper further considers detector design one step after the attack occurs, thereby enhancing transient detection performance against stealthy BIAs.

The rest of the paper is organized as follows. 
The problem formulation is introduced in Section~\ref{sec:problem description}. Section~\ref{sec: main results} presents the design methods for the BIA detectors. 
Section~\ref{sec: algorithms} provides the algorithms and the residual evaluation approach.
The proposed approaches are applied to a thermal system in Section~\ref{sec:simulation} to illustrate their effectiveness.
Section~\ref{sec:conslusion} concludes the paper and outlines future research directions. 
To improve the flow of the paper, some technical proofs are relegated to Appendix~\ref{sec: proof}.

\paragraph{\bf Notation} Sets~$\N$,~$\R~(\R_+)$, and~$\R^n$ denote \mbox{non-negative} integers, (positive) reals, and the space of~$n$-dimensional real vectors, respectively. 
The identity matrix with an appropriate dimension is denoted by~$I$. 
For a matrix~$A$, $\rho(A)$ denotes its spectral radius, i.e., the largest magnitud of all eigenvalues of~$A$,
and $A^{\top}$ denotes the transpose of~$A$.
We use~$A \succ 0 ~(\succeq 0)$ to denote a positive (semi-)definite matrix.
The set of symmetric positive (semi-)definite matrices is denoted by~$\mathbb{S}^{n}_{\succ 0}$ (~$\mathbb{S}^{n}_{\succeq 0}$).
The symbol~$*$ is used to denote the off-diagonal elements in symmetric matrices to avoid clutter.
For a matrix pair~$(A,B)$, its generalized eigenvalues~$\sigma_i$ satisfy the equality~$A \nu = \sigma_i B \nu$ for non-zero vector~$\nu$.
The smallest generalized eigenvalue is denoted by~$\sigma_{\min}(A,B) = \min_{\nu \neq 0, \nu^{\top} B \nu >0} \frac{\nu^{\top} A \nu}{\nu^{\top} B \nu}$.
The quadratic norm of vector~$v$ with respect to~$W \succeq 0$ is defined as~$\|v\|^2_W = v^{\top} W v$.

\section{Model Description and Problem Statement}\label{sec:problem description}
This section first presents the system model considered in this work. The formulation of BIAs is then introduced, followed by their impact analysis, stealthiness characterization, and derivation of worst-case BIAs. 
Finally, the observer design problem against BIAs is formulated.

\subsection{System model}
Consider the following discrete-time LTI system 
\begin{equation}\label{eq:SS model}
    \left\{ \begin{array}l
         x(k+1) = Ax(k) + Bu(k) + B_{\omega} \omega(k)  \\
         y(k) = Cx(k)  + D_{\omega} \omega(k) + y_a(k),
    \end{array}
    \right.
\end{equation}
where~$x(k) \in \R^{n_x}$,~$u(k) \in \R^{n_u}$, and~$y(k) \in \R^{n_y}$ are the state, control input, and output, respectively.
The noise~$\omega(k) \in \R^{n_\omega}$ is assumed to be independently and identically distributed Gaussian white noise with zero mean, i.e.,~$\omega(k) \sim \mathcal{N}(0,I_{n_\omega})$.
The signal~$y_a$ denotes the sensor attack.
The matrices $A, B$, and $C$ are known with appropriate dimensions and are determined by the plant dynamics and structure. 
The matrices $B_{\omega}$ and $D_{\omega}$ characterize the effect of noise on the system.

\begin{assumption}\label{As: observability}
    The pair~$(A,C)$ is detectable and the pair~$(A,B_\omega)$ is stabilizable.
\end{assumption}

We employ an observer to monitor the status of the system
\begin{align}\label{eq: observer}
     \left\{ \begin{array}l
         \hat{x}(k+1) = A\hat{x}(k) + Bu(k) + L(y(k)-\hat{y}(k))  \\
         \hat{y}(k) = C\hat{x}(k) ,
    \end{array}
    \right.
\end{align}
where~$\hat{x}(k) \in \R^{n_x}$ and~$\hat{y}(k) \in \R^{n_y}$ are estimates of the state and output, respectively. 
The observer gain~$L \in \R^{n_x \times n_y}$ is the parameter to be determined.
Define the state estimation error as~$\tilde{x}(k) = x(k)-\hat{x}(k)$, whose dynamics are given by 
\begin{align}\label{eq: error_dynamics}
         \left\{ \begin{array}l
         \tilde{x}(k+1) = (A-LC)\tilde{x}(k) +(B_{\omega}-LD_{\omega})\omega(k) -L y_a(k)  \\
         r(k) = C\tilde{x}(k) + D_{\omega} \omega(k) + y_a(k),
    \end{array}
    \right.
\end{align}
where $r(k) = y(k)-\hat{y}(k)$ is the residual.
Under Assumption~\ref{As: observability}, there always exists an observer gain $L$ such that $A-LC$ is Schur stable. In the absence of attacks or other abnormal signals, $r(k)$ fluctuates around zero due to noise. When an abnormal signal occurs and induces a distinguishable change in $r(k)$, the abnormality can be detected.

\subsection{BIA characterization}\label{Subsec: BIA description} 
In this subsection, we provide a detailed description of the attack signal~$y_a$. Consider the BIA model characterized by a step signal as follows
\begin{align}\label{eq: attack_description}
    y_a(k) = D_a a(k) 
    =\left\{ \begin{array}{ll}
        0,  & k < k_0,   \\
        D_a \bar{a}, & k \geq k_0, 
    \end{array} 
    \right.
\end{align}
where $\bar{a} \in \mathbb{R}^{n_a}$ is a constant attack vector and~$D_a \in \R^{n_y \times n_a}$ characterizes the set of unsecured sensors.
In particular, let $j_1< j_2 < \dots <j_{n_a}$~be the indices of sensors that are potentially compromised. 
Then, the elements $(j_1,1),(j_2,2),\dots,(j_{n_a},n_a)$ of~$D_a$ are set to one, while all other entries are zero.
We further assume that attacks happen at the time instant~$k_0$ and that the error dynamics~\eqref{eq: error_dynamics} have reached steady state by then.

\subsubsection{Impact of BIAs}
BIAs operate by injecting the constant bias~$\bar{a}$ into the system to degrade system performance while remaining as stealthy as possible. 
To quantify such an impact, we introduce the following weighted norm
\begin{align}\label{eq: Weighted_attack}
    \left\|D_a\bar{a}\right\|^2_{W} = \bar{a}^{\top} D^{\top}_a W D_a \bar{a},
\end{align}
where $W$ is a symmetric positive semidefinite matrix. 
For illustration, consider an output feedback controller~$u=Ky$ in~\eqref{eq:SS model}.
In the worst-case scenario, the attacker possesses perfect knowledge of the system and seeks to shift the steady state of the system away from its equilibrium. 
From the closed-loop dynamics, the steady state~$x_{a,\infty}$ induced by the attack vector~$\bar{a}$ is derived as
$x_{a,\infty} = (I-(A+BKC))^{-1}BK D_a \bar{a} = G_{y_a x}  D_a \bar{a}$,
where $G_{y_a x}$ is the steady-state gain from~$y_a$ to $x$.
In this setting, the weighting matrix in~\eqref{eq: Weighted_attack} is chosen as~$W = G^{\top}_{y_a x}G_{y_a x}$.
Note that the controller is introduced only as an illustrative example. The proposed framework itself is not restricted to static \mbox{output-feedback} control and can be applied to other controller architectures as long as $W$ can be specified independently of~$L$.
Otherwise, the resulting impact constraint may become highly nonlinear in~$L$.
Additionally, the choice of~$W$ can be adjusted to different objectives, such as increasing the mean-square estimation error of system states, or degrading the control performance of a linear quadratic Gaussian regulator.

\subsubsection{Stealthiness characterization of BIAs}
We first introduce the definition of KLD, which quantifies the discrepancy between two probability distributions. 

\begin{definition}[KLD~\cite{kullback1997information}]\label{def: KLD} 
Let $p$ and $q$ be two random vectors with probability density functions~$f_p(\cdot)$ and $f_q(\cdot)$, respectively.
Then, the KLD between $p$ and $q$ is 
\begin{align*}
    \mathcal{D}(p || q) = \int_{\{\zeta : f_p(\zeta)>0 \}} f_p(\zeta) \log \frac{f_p(\zeta)}{f_q(\zeta)} \diff \zeta .
\end{align*}
The KLD is non-negative and, in general, non-symmetric. 
A smaller KLD indicates greater similarity between two distributions. It equals zero only for identical distributions. 
\end{definition}

Based on the error dynamics~\eqref{eq: error_dynamics}, the estimation error~$\tilde{x}$ in the attack-free case satisfies
~$\lim_{k \rightarrow \infty} \tilde{x}(k) \sim \mathcal{N}(0,\Sigma_{\tilde{x}})$,
where $\Sigma_{\tilde{x}}$ is the solution to the following discrete-time Lyapunov equation 
\begin{align}\label{eq: DT_Lyapunov}
    \Sigma_{\tilde{x}} = (A-LC)\Sigma_{\tilde{x}}(A-LC)^{\top} + (B_\omega - L D_\omega)  (B_\omega - L D_\omega)^{\top}.
\end{align}
We further define the attack-free residual as~$r_{\omega}$, which satisfies
~$\lim_{k \rightarrow \infty} r_{\omega}(k) \sim \mathcal{N}(0,\Sigma_{r_{\omega}})$ 
with $\Sigma_{r_{\omega}} = C\Sigma_{\tilde{x}}C^{\top} + D_{\omega} D^{\top}_{\omega}$. 
Throughout this work, we assume that $\Sigma_{r_{\omega}}$ is positive definite, which ensures that the KLD expression in~\eqref{eq: BIA_KLD} is well defined.
This assumption is mild and is consistent with common normalized formulations in the literature~\cite{guo2018worst,zhang2019optimal}.

Let $r_a$ denote the residual in the presence of BIAs. 
By the superposition principle, $r_a$ comprises two components, i.e.,~$r_a = r_{\omega} + \bar{r}_{a}$, where~$r_{\omega}$ has been specified above, and~$\bar{r}_{a}$ is the deterministic contribution from the attack~$\bar{a}$.
Recalling that the attack is assumed to start at~$k_0$ with the system initially in the steady state,~$\bar{r}_a$ is derived as
\begin{align}\label{eq: attacked_r}
    \bar{r}_a(k) = \left(I- C\sum^{k-1}_{l=0} (A-LC)^{k-l-1} L \right) D_a \bar{a} 
                 = \Phi(k,L) D_a \bar{a},
\end{align}
with $\Phi(0,L) = I$.
We set~$k_0=0$ here for notational convenience.
Therefore, the residual $r_a$ under attack at the $k$-th step follows~$r_a(k) \sim \mathcal{N}(\bar{r}_a(k),\Sigma_{r_{\omega}})$.
The KLD between~$r_a$ and~$r_{\omega}$ that quantifies the stealthiness of $\bar{a}$ at time $k$ can be computed by
\begin{align}\label{eq: BIA_KLD}
\begin{split}
    \mathcal{D}(r_a(k) || r_{\omega}(k)) 
    &= \frac{1}{2} \bar{r}^{\top}_a(k) \Sigma_{r_{\omega}}^{-1}(L) \bar{r}_a(k)\\
    &= \frac{1}{2} \bar{a}^{\top} D^{\top}_a \Phi^{\top}(k,L) \Sigma_{r_{\omega}}^{-1}(L) \Phi(k,L) D_a \bar{a}.
\end{split}
\end{align}

While some studies adopt horizon-based KLD to characterize attack detectability~\cite{bai2017data,zhang2019optimal}, these approaches require long observation horizons and high-dimensional statistical modeling, which limits their suitability for detector synthesis. 
In contrast, this work adopts the instantaneous KLD in~\eqref{eq: BIA_KLD} as a measure of attack detectability. Under Gaussian assumptions, it is closely related to the Mahalanobis distance and classical $\chi^2$ tests (outlined in Section~\ref{sec: algorithms}), thereby providing a natural link to \mbox{residual-based} detection. 
Moreover, the single-instant consideration enables tractable detection observer gain design.

\subsubsection{Worst-case BIA design}
Given the BIA impact characterized by the weighted norm in~\eqref{eq: Weighted_attack} and the stealthiness quantified by the KLD in~\eqref{eq: BIA_KLD}, the worst-case BIA vector denoted by~$\bar{a}^*$ is obtained as the solution to the following optimization problem:
\begin{align}\label{eq: Opt_a}
\begin{split}
   J_k(L) =  \min_{\bar{a} \in \R^{n_a}} ~&\frac{1}{2} \bar{a}^{\top} D^{\top}_a \Phi^{\top}(k,L) \Sigma_{r_{\omega}}^{-1}(L) \Phi(k,L) D_a \bar{a} \\
    \textup{s.t.} ~&\bar{a}^{\top} D^{\top}_a W D_a \bar{a} \geq \epsilon,
\end{split}
\end{align}
where $J_k(L)$ represents the smallest KLD value at time $k$ and~$\epsilon \in \R_+$ is the lower bound on the BIA impact. 
Since both the objective and the impact constraint are homogeneous quadratic forms, we can restrict the search to the boundary, and an optimal solution can be chosen to satisfy the equality constraint, i.e.,~$\bar{a}^{\top} D^{\top}_a W D_a \bar{a} = \epsilon$. Without loss of generality, we set~$\epsilon=1$. 
For notational simplicity, define $\Psi(k,L)=\frac{1}{2} D^{\top}_a \Phi^{\top}(k,L) \Sigma_{r_{\omega}}^{-1}(L) \Phi(k,L) D_a$ and $\Gamma = D^{\top}_a W D_a$.

Note that for a fixed time instant~$k$ and a known observer gain~$L$, the worst-case BIA~$\bar{a}^*$ to the optimization problem~\eqref{eq: Opt_a} can be obtained by computing the smallest generalized eigenvalue of the matrix pair~$(\Psi(k,L),\Gamma)$~\cite{ghojogh2019eigenvalue}, as denoted by~$\sigma_{\min}(\Psi(k,L),\Gamma)$. 
In what follows, we outline the derivation of the worst-case BIA in two cases: when $\Gamma$ is positive definite and when it is positive semidefinite. 
\begin{enumerate}
    \item \textbf{Matrix $\Gamma \succ 0$.} 
    The worst-case BIA $\bar{a}^*$ is given by
    \begin{align*}
        \bar{a}^* = \pm \frac{\sqrt{\epsilon}}{\| v_{\text{pd}}^* \|_{\Gamma}} v^*_{\text{pd}},
    \end{align*}
    where~$v_{\text{pd}}^*$ is the eigenvector corresponding to $\sigma_{\min}(\Psi(k,L),\Gamma)$.
 
    \item \textbf{Matrix $\Gamma \succeq 0$ is singular.} Let the eigenvalue decomposition of $\Gamma$ be
    \begin{align*}
        \Gamma =
        \begin{bmatrix}
        U_r & U_n
        \end{bmatrix}
        \begin{bmatrix}
        \Lambda_\Gamma & 0\\
        0 & 0
        \end{bmatrix}
        \begin{bmatrix}
        U_r^\top\\
        U_n^\top
        \end{bmatrix},
    \end{align*}
    where $\Lambda_\Gamma\succ0$ contains the nonzero eigenvalues of $\Gamma$, and $U_r \in \mathbb{R}^{n_a \times p}$ and $U_n \in \mathbb{R}^{n_a \times q}$ span the range space and the null space of $\Gamma$, respectively. Then any attack vector can be decomposed as $\bar a = U_r b + U_n c$. Under this decomposition, the optimization problem~\eqref{eq: Opt_a} is written as  
    \begin{align*} \min_{b\in\mathbb{R}^{p},c\in\mathbb{R}^{q}}~b^\top\Psi^{rr}_{k,L}b + 2b^\top\Psi^{rn}_{k,L}c + c^\top\Psi^{nn}_{k,L}c \quad\text{s.t.}~b^\top \Lambda b \geq \epsilon
    \end{align*}
    where $\Psi^{rr}_{k,L}=U_r^\top\Psi(k,L)U_r,
    ~\Psi^{rn}_{k,L}=U_r^\top\Psi(k,L)U_n,
    ~\Psi^{nr}_{k,L}=U_n^\top\Psi(k,L)U_r,
    \Psi^{nn}_{k,L}=U_n^\top\Psi(k,L)U_n.$

    Note that $c$ in the null space of $\Gamma$ does not affect the impact constraint, but it may still affect the value of the KLD objective through $2b^\top\Psi^{rn}_{k,L}c + c^\top\Psi^{nn}_{k,L}c$.
    For a fixed $b$, the minimization of this part with respect to $c$ is bounded from below if (i) $\Psi^{nn}_{k,L}\succeq0$ and (ii) $\Psi^{nr}_{k,L}b\in \operatorname{Range}(\Psi^{nn}_{k,L})$.
    The first condition is straightforward. 
    The second condition is relevant when $\Psi^{nn}_{k,L}$ is singular.
    Indeed, if $\Psi^{nr}_{k,L}b\notin \operatorname{Range}(\Psi^{nn}_{k,L})$, then there exists a vector $z\in\operatorname{ker}(\Psi^{nn}_{k,L})$ such that $z^\top\Psi^{nr}_{k,L}b\neq0$.
    By choosing $c=tz$, where $t\in\mathbb R$ is a scalar, we obtain
    $2b^\top\Psi^{rn}_{k,L}c+c^\top\Psi^{nn}_{k,L}c
    = 2t b^\top\Psi^{rn}_{k,L}z \neq 0$, which can be made arbitrarily negative by choosing $t$ with a suitable sign and letting $|t|\to\infty$.

    In this problem, the two boundedness conditions, i.e., $\Psi^{nn}_{k,L}\succeq0$ and $\Psi^{nr}_{k,L}b\in \operatorname{Range}(\Psi^{nn}_{k,L})$, are automatically satisfied because $\Psi(k,L)$ is symmetric positive semidefinite.
    Hence, it holds that $\operatorname{Range}(\Psi^{nr}_{k,L})\subseteq \operatorname{Range}(\Psi^{nn}_{k,L})$\cite[Theorem 1.19]{zhang2005schur} and $\Psi^{nr}_{k,L} b\in \operatorname{Range}(\Psi^{nn}_{k,L})$ for any $b$.

    As a result, $c^\star=-\Psi^{nn \dagger}_{k,L} \Psi^{nr}_{k,L} b$ is one minimizer for a fixed~$b$, where $\Psi^{nn \dagger}_{k,L}$ denotes the Moore--Penrose inverse. The reduced optimization problem becomes
\begin{align*}
    \min_{b\in\mathbb R^p}
    b^\top \left(
    \Psi^{rr}_{k,L} -
    \Psi^{rn}_{k,L}\Psi^{nn \dagger}_{k,L} \Psi^{nr}_{k,L}
    \right)b,
    \quad
    \text{s.t. }
    b^\top\Lambda b\ge \epsilon.
\end{align*}
The optimal attack vector is then obtained as
$\bar a^\star = U_r b^\star + U_n c^\star
    = \left(U_r - U_n \Psi^{nn \dagger}_{k,L} \Psi^{nr}_{k,L} \right)  b^\star$,
where
$b^\star =
\pm \frac{\sqrt{\epsilon}}
{\|v^\star_{\rm psd}\|_{\Lambda}}
v^\star_{\rm psd}$,
and $v^\star_{\rm psd}$ is the generalized eigenvector corresponding to the smallest generalized eigenvalue 
$\sigma_{\min}\left(
\Psi^{rr}_{k,L} - \Psi^{rn}_{k,L} \Psi^{nn \dagger}_{k,L} \Psi^{nr}_{k,L}, \Lambda
\right)$.
\end{enumerate}
For further details on the generalized eigenvalue optimization problem and its application to BIA design, we refer readers to~\cite{ghojogh2019eigenvalue,milovsevivc2017analysis,teixeira2015secure}.

\subsection{Problem statement}
To enhance the detectability of worst-case BIAs, we maximize the smallest KLD value~$J_k(L)$ in~\eqref{eq: Opt_a} by optimizing the observer gain $L$. 
Note that $J_k(L)$ is time-varying, which leads to different objective functions at different time steps~$k$. 
Nonetheless, we focus on several specific instants because detection performance need not be guaranteed at every time step and $J_k(L)$ contains higher-order terms in~$L$ (for $1< k <\infty$), resulting in increased computational complexity. 
Specifically, (1) \textbf{transient} detection performance can be improved by considering attack onset at~$k=0$ or one step after attack occurrence at $k=1$, and (2) \textbf{steady-state} detection performance can be enhanced by considering the limit case~$k \to \infty$.
Then, we formulate the problem addressed in this paper.

\begin{problem}[Optimal detection observer gain design against BIAs]\label{pro: Observer_design}
Consider the error dynamics~\eqref{eq: error_dynamics} and the worst-case BIA in~\eqref{eq: Opt_a}. 
For $k=0,~1,~\infty$, design the observer gain~$L$ through the following optimization problem 
    \begin{align}\label{eq: Pro_OptBIADet}
        \max_{L \in \R^{n_x \times n_y}} \{J_k(L): \rho(A-LC)<1 \} ,
    \end{align}
where $\rho(A-LC)<1$ is introduced to ensure the stability of the closed-loop estimator.   
\end{problem}

Note that Problem~\eqref{eq: Pro_OptBIADet} is non-convex and highly nonlinear because both the residual covariance matrix~$\Sigma_{r_{\omega}}(L)$ and the transition matrix $\Phi(k,L)$ depend nonlinearly on $L$. The stability constraint introduces additional nonlinearity into the problem. 

\begin{remark}[Rationale for constant BIA model]
    Besides stealthiness, BIAs are considered for several additional reasons: 
    (1) The constant-BIA model enables tractable observer synthesis with theoretical guarantees; 
    (2) Although the design is based on constant biases, it can also improve the detection of slowly increasing bias attacks, as illustrated in Section~\ref{sec:simulation}, as well as more complicated sum-of-exponentials attacks~\cite{sandberg2021stealthy}, whose responses resemble slowly varying biases; 
    (3) BIAs are relevant to systems with integrators, where their effect on the residual may vanish while their impact on the system persists, which is known as integrator vulnerability~\cite{tosun2025kullback}.
\end{remark}


\section{Main results}\label{sec: main results}
We first introduce several lemmas. 
Leveraging these results, problem~\eqref{eq: Pro_OptBIADet} for $k=1$ and $k=\infty$ is formulated into more tractable bi-convex and LMI approximations.
We begin with a Lagrange-dual reformulation of~\eqref{eq: Pro_OptBIADet}.

\begin{lemma}[Reformulation of the observer design problem via Lagrange dual]\label{lem: reform pro}
    The optimization problem~\eqref{eq: Pro_OptBIADet} used for observer design can be reformulated as
    \begin{subequations}\label{eq: Reform_Pro_OptBIADet}
        \begin{align}
        \max_{L \in \R^{n_x \times n_y},~\lambda \geq 0} ~&\lambda \notag\\
        \textup{s.t.} ~&\frac{1}{2} D^{\top}_a \Phi^{\top}(k,L) \Sigma_{r_{\omega}}^{-1}(L) \Phi(k,L) D_a - \lambda D^{\top}_a W D_a  \succeq 0, \label{eq: Reform_Pro_OptBIADet 1}\\
        & \rho(A-LC)<1. \label{eq: Reform_Pro_OptBIADet 2}
    \end{align}
    \end{subequations}
\end{lemma}
\begin{proof}
    The Lagrangian of~\eqref{eq: Opt_a} is first obtained as
    \begin{align*}
        \mathcal{L}(\bar{a},\lambda) = &\frac{1}{2} \bar{a}^{\top} D^{\top}_a \Phi^{\top}(k,L) \Sigma_{r_{\omega}}^{-1}(L) \Phi(k,L) D_a \bar{a} + \lambda (1-\bar{a}^{\top} D^{\top}_a W D_a \bar{a}),
    \end{align*}
    where~$\lambda \geq 0$ is the Lagrange multiplier.
    The corresponding Lagrange dual function~$g(\lambda)$, which is the minimum value of the Lagrangian~$\mathcal{L}(\bar{a},\lambda)$ over~$\bar{a}$, is derived as
    \begin{align*}
        g(\lambda) = \min_{\bar{a} \in \R^{n_a}} \mathcal{L}(\bar{a},\lambda) 
        =\left\{\begin{array}{ll}
          \lambda,   & \frac{1}{2}  D^{\top}_a \Phi^{\top}(k,L) \Sigma_{r_{\omega}}^{-1}(L) \Phi(k,L)  D_a  \\ 
          &- \lambda  D^{\top}_a W D_a \succeq 0,\\
           -\infty,   & \textup{otherwise.} 
        \end{array} \right.
    \end{align*}
    Since the impact constraint admits a strictly feasible point for~$\Gamma \neq 0$, strong duality follows from the S-lemma for an optimization problem with a quadratic objective function and one quadratic inequality constraint~\cite[Chapter 5.2]{boyd2004convex}.
    Therefore, the obtained Lagrange dual optimization problem has the optimal value as~\eqref{eq: Opt_a}, i.e.,
    \begin{align*}
        J_k(L) = \max_{\lambda \geq 0} ~\lambda 
        \quad \textup{s.t.} ~\frac{1}{2}  D^{\top}_a \Phi^{\top}(k,L) \Sigma_{r_{\omega}}^{-1}(L) \Phi(k,L)  D_a - \lambda  D^{\top}_a W D_a \succeq 0.
    \end{align*}
    Taking the constraint~$\rho(A-LC)<1$ into account, we reformulate the max-min optimization problem~\eqref{eq: Pro_OptBIADet} into~\eqref{eq: Reform_Pro_OptBIADet} by jointly optimizing over~$L$ and~$\lambda$. This completes the proof.
\end{proof}

The inverse of the covariance matrix~$\Sigma_{r_{\omega}}^{-1}(L)$ in constraint~\eqref{eq: Reform_Pro_OptBIADet 1} poses a great challenge in solving the optimization problem. 
To address this issue, we introduce the following lemma that provides a lower relaxation of~$\Sigma_{r_{\omega}}^{-1}(L)$ while ensuring the stability condition. 

\begin{lemma}[Relaxation of the inverse covariance matrix]\label{lem: Alt_InvMat}
    Consider the error dynamics~\eqref{eq: error_dynamics} and the inverse of the covariance matrix~$\Sigma_{r_{\omega}}^{-1}(L)$. 
    The following statements hold true:  
    \begin{enumerate}
        \item $\rho(A-LC)<1$,
        \item $\Sigma^{-1}_{r_\omega}(L) \succeq (C P^{-1} C^{\top} + D_{\omega} D_{\omega}^{\top} )^{-1} \succeq \tilde{Z}$,
    \end{enumerate}    
    if there exist symmetric positive definite matrices~$P$ and $\tilde{Z}$ such that the following inequalities are satisfied: 
    \begin{align}\label{eq: LMIs}
    \begin{bmatrix}
        P &P(A-LC) &P(B_{\omega} - L D_{\omega})\\
        * &P       &0\\
        * &*       &I
    \end{bmatrix} \succ 0, 
    \quad
    \begin{bmatrix}
        \tilde{Z} &\tilde{Z}C &\tilde{Z}D_{\omega}\\
        * &P &0\\
        * &* &I
    \end{bmatrix} \succ 0.
    \end{align} 
\end{lemma}
\begin{proof}
    The first inequality in~\eqref{eq: LMIs} implies that 
    \begin{align*}
        \begin{bmatrix}
            P &P(A-LC) \\
            * &P
        \end{bmatrix} \succ 0,
    \end{align*}
    which is equivalent to~$P-(A-LC)^{\top}P(A-LC) \succ 0$ according to Schur complement.
    This ensures the stability of the error dynamics~\eqref{eq: error_dynamics}, i.e., $\rho(A-LC)<1$.

    For the second statement, we frst show that~$P^{-1} \succ \Sigma_{\tilde{x}}$ if the first inequality in~\eqref{eq: LMIs} holds.
    By utilizing the Schur complement, we have 
    \begin{align*}
    P - \begin{bmatrix}
        P(A-LC) &P(B_{\omega} - L D_{\omega})
    \end{bmatrix}
    \begin{bmatrix}
        P &0\\ * &I
    \end{bmatrix}^{-1}
    \begin{bmatrix}
        (A-LC)^{\top} P \\ (B_{\omega} - L D_{\omega})^{\top} P
    \end{bmatrix} \succ 0.
    \end{align*}
    Pre- and post-multiplying the above inequality by~$P^{-1}$ and reversing the sign of both sides lead to
    \begin{align*}
        (A-LC)P^{-1}(A-LC)^{\top} - P^{-1} + (B_{\omega} - L D_{\omega})(B_{\omega} - L D_{\omega})^{\top} \prec 0.
    \end{align*}
    Recall that the covariance matrix~$\Sigma_{\tilde{x}}$ is the solution to the discrete-time Lyapunov equation in~\eqref{eq: DT_Lyapunov}.
    It holds that~$(B_\omega - L D_\omega)  (B_\omega - L D_\omega)^{\top} =\Sigma_{\tilde{x}}-(A-LC)\Sigma_{\tilde{x}}(A-LC)^{\top}$.
    By substituting this relation into the above inequality, we have
    \begin{align*}
        (A-LC)(P^{-1}-\Sigma_{\tilde{x}})(A-LC)^{\top} - (P^{-1} - \Sigma_{\tilde{x}}) \prec 0.
    \end{align*}
    Since $A-LC$ is Schur, according to the discrete-time Lyapunov stability results in~\cite[Theorem 8.4]{hespanha2018linear},~$P^{-1}-\Sigma_{\tilde{x}}$ should be a symmetric positive definite matrix, i.e.,~$P^{-1} \succ \Sigma_{\tilde{x}}$, which further implies that  
    $\Sigma^{-1}_{r_\omega}(L) = (C \Sigma_{\tilde{x}} C^{\top} + D_{\omega} D_{\omega}^{\top})^{-1} \succeq (C P^{-1} C^{\top} + D_{\omega} D_{\omega}^{\top} )^{-1}$.
    
    To show~$(C P^{-1} C^{\top} + D_{\omega} D_{\omega}^{\top} )^{-1} \succeq \tilde{Z}$, by pre- and post-multiplying the second inequality in~\eqref{eq: LMIs} by~$\text{diag}(\tilde{Z}^{-1},I,I)$, the inequality is readily from the Schur complement, which is
    \begin{align*}
        \tilde{Z}^{-1} - \begin{bmatrix}
            C &D_\omega
        \end{bmatrix}
        \begin{bmatrix}
            P &0 \\ * &I
        \end{bmatrix}^{-1}
        \begin{bmatrix}
            C^{\top} \\ D^{\top}_{\omega}
        \end{bmatrix} 
        = \tilde{Z}^{-1} - (C P^{-1} C^{\top} + D_{\omega} D^{\top}_{\omega}) \succ 0.
    \end{align*}
    This completes the proof.
\end{proof}



We are now ready to present the main results of the paper.

\subsection{KLD-based observer design for attack-onset detectability}
We first consider the case~$k=0$ in~\eqref{eq: Pro_OptBIADet}, which aims to enhance the detectability of worst-case BIAs at the attack onset.
According to~\eqref{eq: attacked_r}, the contribution of BIAs to the residual at~$k=0$ is~$\bar{r}_a(0) = \Phi(0,L) D_a \bar{a}$, where~$\Phi(0,L)=I$ is the identity matrix.
Consequently, the reformulated optimization problem~\eqref{eq: Reform_Pro_OptBIADet} for observer design becomes
\begin{align}\label{eq: OptPro_KLD0}
\begin{split}
     J_0(L^*_0) = &\max_{L\in \R^{n_x \times n_y}, ~\lambda \geq 0} ~\lambda \\
        \textup{s.t.} ~&~\frac{1}{2}  D^{\top}_a  \Sigma_{r_{\omega}}^{-1}(L)   D_a - \lambda  D^{\top}_a W D_a \succeq 0~\textup{and}~\rho(A-LC)<1,
\end{split}
\end{align}
where $L^*_0$ is the optimal solution to~\eqref{eq: OptPro_KLD0} and~$J_0(L^*_0)$ denotes the optimal objective value.
The following proposition shows that the Kalman filter is optimal with respect to onset detectability.

\begin{proposition}[Optimality of the Kalman filter at attack onset]\label{prop: kld0}
    For the attack-onset case, the Kalman filter is the optimal solution to~\eqref{eq: OptPro_KLD0}, irrespective of the set of compromised sensors characterized by~$D_a$ and the impact matrix $W$, provided that $D_a^\top W D_a\neq0$.
\end{proposition}

\begin{proof}
    Define the sets of stabilizing observer gains and feasible attacks as
    \begin{align*}
        \mathbb{L}:=
    \{L\in\mathbb{R}^{n_x\times n_y}:\rho(A-LC)<1\}, ~\mathcal A=\{\bar a \in \R^{n_a}:
    \bar a^\top D_a^\top W D_a\bar a\ge1\},
    \end{align*}
    respectively. 
    At the attack onset, the worst-case detectability in~\eqref{eq: Opt_a} becomes
    \begin{align*}
        J_0(L)=\min_{\bar a\in\mathcal A}\frac{1}{2}\bar a^\top D_a^\top \Sigma_{r_\omega}^{-1}(L)D_a\bar a.
    \end{align*}
    The Kalman filter gain $L_{\rm kal}$ minimizes the residual covariance, such that
    $\Sigma_{r_\omega}(L_{\rm kal}) \preceq \Sigma_{r_\omega}(L), ~\forall L\in\mathbb{L}$.
    Since $\Sigma_{r_\omega}(L)\succ0$, it follows that
    $\Sigma_{r_\omega}^{-1}(L_{\rm kal}) \succeq \Sigma_{r_\omega}^{-1}(L)$.
    Therefore, for any feasible attack vector~$\bar a$,
    \begin{align*}
        \frac{1}{2}\bar a^\top D_a^\top \Sigma_{r_\omega}^{-1}(L_{\rm kal})D_a\bar a \ge \frac{1}{2}\bar a^\top D_a^\top \Sigma_{r_\omega}^{-1}(L)D_a\bar a.
    \end{align*}

    Since the attack-impact constraint is independent of $L$, the feasible attack set $\mathcal{A}$ is the same for all~$L$, regardless of whether~$D_a^\top W D_a$ is positive definite or positive semidefinite and singular.
    Then, taking the minimum over the same feasible attack set~$\mathcal{A}$ yields
    $J_0(L_{\rm kal})\ge J_0(L), ~\forall L\in\mathbb{L}$.
    Thus, the Kalman filter is optimal to~\eqref{eq: OptPro_KLD0},  irrespective of $D_a$ and $W$.
    This completes the proof.
\end{proof}

\subsection{KLD-based observer design for one-step detectability}\label{subsec: KLD1}
Setting $k=1$ in~\eqref{eq: Pro_OptBIADet} can enhance the detectability of \mbox{worst-case} BIAs one step after attack occurrence.
This is particularly relevant for systems subject to integrator vulnerability, where the residual effect of a BIA vanishes at steady state, making transient detection important.  
According to~\eqref{eq: attacked_r}, the contribution of BIAs to the residual at~$k=1$ becomes~$\bar{r}_a(1) = (I-CL) D_a \bar{a}$, where~$\Phi(1,L)=I-CL$.
Building on Lemma~\ref{lem: reform pro}, the observer design problem for the one-step case becomes
\begin{align}\label{eq: OptPro_KLD1}
    J_1(L^*_1) &= \max_{L\in \R^{n_x \times n_y}, ~\lambda \in \R_+} ~\lambda \notag\\
        \textup{s.t.} ~&\frac{1}{2}  D^{\top}_a (I-CL)^{\top} \Sigma_{r_{\omega}}^{-1}(L) (I-CL)  D_a - \lambda  D^{\top}_a W D_a \succeq 0 
        ~\textup{and}~\rho(A-LC)<1,
\end{align}
where $L_1^*$ is the optimal solution to~\eqref{eq: OptPro_KLD1} and $J_1(L^*_1)$ is the corresponding optimal objective value.
The matrix inequality constraint in~\eqref{eq: OptPro_KLD1} is complicated due to the composite term $D^{\top}_a (I-CL)^{\top} \Sigma_{r_{\omega}}^{-1}(L) (I-CL)  D_a$. 
To address this issue, the following theorem provides a bi-convex approximation of~\eqref{eq: OptPro_KLD1}.

\begin{theorem}[Bi-convex approximation for one-step detectability]\label{thm: KLD1_BMI}
    Consider the error dynamics~\eqref{eq: error_dynamics} in the presence of BIAs. 
    The reformulated optimization problem~\eqref{eq: OptPro_KLD1} of the detection observer design problem~\eqref{eq: Pro_OptBIADet} at~$k=1$ can be approximated through the following bi-convex optimization problem:
    \begin{align}\label{eq: BMIs_KLD1}
    \max ~&\lambda \notag\\
     \textup{s.t.} ~&L\in \R^{n_x \times n_y}, ~\lambda \geq 0, P \in \mathbb{S}^{n_x}_{\succ 0}, ~\tilde{Z} \in \mathbb{S}^{n_y}_{\succ 0}, ~Y \in \mathbb{R}^{n_y \times n_a}, ~\eqref{eq: LMIs},  \notag  \\
          &\begin{bmatrix}
         \varphi(Y,L) - \lambda D_a^{\top}WD_a &Y^{\top} \\
         Y &\frac{1}{2} \tilde{Z}
    \end{bmatrix} \succeq 0, 
    \end{align}
    where $\varphi(Y,L) = Y^{\top}(I-CL)D_a + D_a^{\top}(I-CL)^{\top}Y$.
    The optimal objective value $\lambda^*$ satisfies~$\lambda^* \leq J_1(L_1^*)$.
\end{theorem}

\begin{proof}
    Since~$\frac{1}{2}((I-CL)D_a - 2\Sigma_{r_\omega}(L) Y)^{\top}\Sigma^{-1}_{r_{\omega}}(L)((I-CL)D_a - 2\Sigma_{r_{\omega}}(L)Y) \succeq 0$, it holds that
    \begin{align}\label{eq: KLD1_BMI_proof 1}
        \frac{1}{2}D^{\top}_a (I-CL)^{\top}\Sigma^{-1}_{r_{\omega}}(L) (I-CL) D_a 
        \succeq  \varphi(Y,L) - 2 Y^{\top} \Sigma_{r_{\omega}}(L) Y.
    \end{align}
    According to Lemma~\ref{lem: Alt_InvMat}, inequalities~\eqref{eq: LMIs} ensure that $\rho(A-LC)<1$ and  $\Sigma_{r_{\omega}}^{-1}(L) \succ \tilde{Z}$. Together with inequality~\eqref{eq: KLD1_BMI_proof 1}, we have
    \begin{align}\label{eq: KLD1_BMI_proof 2}
        \frac{1}{2}D^{\top}_a (I-CL)^{\top}\Sigma^{-1}_{r_{\omega}}(L) (I-CL) D_a 
        \succeq  \varphi(Y,L) - 2  Y^{\top} \tilde{Z}^{-1} Y.
    \end{align}
    Then, applying the Schur complement to~\eqref{eq: BMIs_KLD1} leads to
    $\varphi(Y,L) - 2Y^{\top} \tilde{Z}^{-1} Y \succeq  \lambda D_a^{\top}WD_a$.
 
    By substituting~\eqref{eq: KLD1_BMI_proof 2} into the above inequality, we have
    \begin{align*}
        \frac{1}{2}  D^{\top}_a (I-CL)^{\top} \Sigma_{r_{\omega}}^{-1}(L) (I-CL)  D_a - \lambda  D^{\top}_a W D_a \succeq 0.
    \end{align*}
    Hence, the conditions in~\eqref{eq: BMIs_KLD1} are sufficient for~\eqref{eq: OptPro_KLD1}, and thus the optimal objective value $\lambda^*$ provides a lower bound on that of~\eqref{eq: OptPro_KLD1}. 
    This completes the proof.  
\end{proof}

The optimization problem~\eqref{eq: BMIs_KLD1} contains bilinear terms~$PL$ and~$Y^{\top}CL$ in its constraints and can be solved using AO and ADMM algorithms.
The characterization of the two algorithms will be outlined in Subsection~\ref{subsec: AO and ADMM}.
Alternatively, the observer design problem~\eqref{eq: OptPro_KLD1} can be further approximated by more tractable LMIs in the following theorem, albeit at the cost of increased conservatism. 
To this end, the impact matrix $W$ is factorized as~$W=R^{\top}_W R_W$.

\begin{theorem}[LMI-based approximation for one-step detectability]\label{thm: KLD1LMI}
    Consider the error dynamics~\eqref{eq: error_dynamics} in the presence of BIAs. 
    Given a scalar~$\gamma \in \R_+$, the reformulated optimization problem~\eqref{eq: OptPro_KLD1} for detection observer design at~$k=1$, can be approximated through the following LMI-based optimization problem:
    \begin{subequations}\label{eq: LMIs_KLD1}
         \begin{align}
            \min ~&\mu  \notag \\
            \textup{s.t.}   ~&\mu \in \mathbb{R}_+, ~P \in \mathbb{S}^{n_x}_{\succ 0}, ~\tilde{Z} \in \mathbb{S}^{n_y}_{\succ 0}, ~G \in \mathbb{R}^{n_x \times n_y}, ~Y \in \mathbb{R}^{n_y \times n_a}, \notag \\
            &\begin{bmatrix}
                P &PA-GC &PB_{\omega}-GD_{\omega}\\
                * &P       &0\\
                * &*       &I
            \end{bmatrix} \succ 0,
            ~\begin{bmatrix}
                \tilde{Z} &\tilde{Z}C &\tilde{Z}D_{\omega}\\
                * &P &0\\
                * &0 &I
            \end{bmatrix} \succ 0, \label{eq: LMIs_KLD1 c1}\\
            &\begin{bmatrix}
                \begin{bmatrix}
                Y^{\top}D_a + D^{\top}_a Y   &D^{\top}_a R^{\top}_W \\
                R_W D_a &2 \mu I
                    \end{bmatrix} 
                &\begin{bmatrix}
                    \Gamma_1 &0\\
                    0        &\mu I
                \end{bmatrix} \\
            * &\begin{bmatrix}
                \Gamma_2 &0 \\
             0 &\mu I
            \end{bmatrix}
            \end{bmatrix}\succeq 0, \label{eq: LMIs_KLD1 c2}
        \end{align}
    \end{subequations} 
    and the involved matrices are
    \begin{align*}
    \Gamma_1 = 
        \begin{bmatrix}
            Y^{\top}C  &D^{\top}_a G^{\top} &Y^{\top}D_\omega 
        \end{bmatrix}, 
    ~\Gamma_2 = \begin{bmatrix}
            (2+1/\gamma)^{-1} P &0 &0\\
            *  &1/\gamma P &0 \\
            *  &* &1/2 I
            \end{bmatrix},
    \end{align*}
    where $\mu = 1/\lambda$, $W=R^{\top}_W R_W$ and the observer gain matrix~$L = P^{-1}G$.
\end{theorem}

\begin{proof}
The proof is relegated to Section~\ref{sec: proof}.
\end{proof}

\subsection{KLD-based observer design for steady-state detectability}\label{subsec: KLDinf}
To enhance the steady-state detectability of worst-case BIAs, we consider $k \to \infty$ in the optimization problem~\eqref{eq: Reform_Pro_OptBIADet} and obtain:
\begin{align}\label{eq: OptPro_KLD_inf}
     J_{\infty}(L^*_{\infty}) &= \max_{L \in \mathbb{R}^{n_x \times n_y}, ~\lambda \geq 0} ~\lambda \notag\\
        \textup{s.t.} ~&\frac{1}{2} D^{\top}_a \Phi(\infty,L)^{\top} \Sigma_{r_{\omega}}(L)^{-1} \Phi(\infty,L) D_a - \lambda D^{\top}_a W D_a \succeq 0
        ~\textup{and}~\rho(A-LC)<1,
\end{align}
where $\Phi(\infty,L) = I-C(I-A + LC)^{-1} L$, $L_{\infty}^*$ is the optimal solution to~\eqref{eq: OptPro_KLD_inf} with $J_{\infty}(L^*_{\infty})$ being the corresponding optimal objective value.
To deal with the non-tractable term~$\Phi(\infty,L)$ in the constraint, we derive an alternative expression of~$\Phi(\infty,L)$ in the following lemma.

\begin{lemma}[Transformation of the steady-state transition matrix]\label{lem: Alternative trans matrix}
    Assume that the system matrix $A$ has no unit eigenvalues and $A-LC$ is Schur. 
    The steady-state transition matrix $\Phi(\infty,L)$ in~\eqref{eq: attacked_r} can be written as~
    \begin{align}\label{eq: alternative expression}
        \Phi(\infty,L) = (I+ML)^{-1},
    \end{align}
    where~$M = C(I-A)^{-1}$.
\end{lemma}
\begin{proof}
The proof is relegated to Section~\ref{sec: proof}.
\end{proof}

The Schur stability of $(A-LC)$ ensures the existence of a steady-state transition matrix $\Phi(\infty,L)$, while the absence of unit eigenvalues in $A$ guarantees the invertibility of $I-A$.
With the alternative expression of the steady-state transition matrix in~\eqref{eq: alternative expression}, we obtain a bi-convex approximation of the optimization problem~\eqref{eq: OptPro_KLD_inf} in the following theorem.

\begin{theorem}[Bi-convex approximation for steady-state detectability]\label{thm: KLDinf BMI}
    Consider the error dynamics~\eqref{eq: error_dynamics} in the presence of BIAs and assume that~$A$ has no unit eigenvalues.
    Utilizing the alternative expression of~$\Phi(\infty,L)$ in~\eqref{eq: alternative expression}, the reformulated optimization problem~\eqref{eq: OptPro_KLD_inf} for detection observer design at~$k=\infty$ can be approximated through the following bi-convex optimization problem:
    \begin{align}
        \max ~&\lambda \notag\\
         \textup{s.t.}~&\lambda \geq 0, L \in \mathbb{R}^{n_x \times n_y}, P \in \mathbb{S}^{n_x}_{\succ 0}, \tilde{Z} \in \mathbb{S}^{n_y}_{\succ 0}, Y \in \mathbb{R}^{n_y \times n_a}, \eqref{eq: LMIs}, \notag \\
         &\begin{bmatrix}
            D^{\top}_a Y + Y^{\top} D_a-\lambda D^{\top}_a W D_a &Y^{\top}(I+ML)\\
            * &\frac{1}{2} \tilde{Z}
        \end{bmatrix} \succeq 0. \label{eq: BMIs_KLD_inf}
    \end{align}  
    The optimal objective value $\lambda^*$ satisfies~$\lambda^* \leq J_{\infty}(L_{\infty}^*)$.
\end{theorem}

\begin{proof} 
    Utilizing the Young's relation in Lemma~\ref{lem: Matrix_Ineq} leads to
    \begin{align}\label{eq: KLDinf_BMI_proof 1}
    &\frac{1}{2}D^{\top}_a (I + ML )^{-\top} \Sigma_{r_{\omega}}(L)^{-1} (I + ML )^{-1} D_a \notag\\ \succeq  &D^{\top}_a Y + Y^{\top} D_a  - 2 Y^{\top} (I+ML) \Sigma_{r_{\omega}}(L) (I+ML)^{\top} Y,
    \end{align}
    where $(I + ML) \Sigma_{r_{\omega}}(L) (I + ML )^{\top}$ is symmetric positive definite as $I+ML$ is invertible.
    Then, applying Schur complement to~\eqref{eq: BMIs_KLD_inf} and considering the inequality $\Sigma_{r_{\omega}}(L)^{-1} \succeq \tilde{Z}$ from \eqref{eq: LMIs}, we have~$D^{\top}_a Y + Y^{\top} D_a  - 2 Y^{\top} (I+ML) \Sigma_{r_{\omega}}(L) (I+ML)^{\top} Y -\lambda D^{\top}_a W D_a \succeq 0$. 
    Together with the inequality~\eqref{eq: KLDinf_BMI_proof 1}, it holds that
    \begin{align*}
        \frac{1}{2}D^{\top}_a (I + ML )^{-\top} \Sigma_{r_{\omega}}(L)^{-1} (I + ML )^{-1} D_a -\lambda D^{\top}_a W D_a \succeq 0.
    \end{align*}
    Therefore, constraints in~\eqref{eq: OptPro_KLD_inf} are satisfied. This completes the proof.
\end{proof}

In addition, an LMI-based approximation of~\eqref{eq: OptPro_KLD_inf} is further provided in the following theorem.

\begin{theorem}[LMI-based approximation for steady-state detectability]\label{thm: KLDinf_LMI}
    Consider the error dynamics~\eqref{eq: error_dynamics} in the presence of BIAs and assume that~$A$ has no unit eigenvalues.
    Given a scalar~$\gamma \in \R_+$ and the alternative expression of~$\Phi(\infty,L)$ in~\eqref{eq: alternative expression}, the reformulated optimization problem~\eqref{eq: OptPro_KLD_inf} for detection observer design at~$k=\infty$ can be approximated through the following LMI-based optimization problem:
    \begin{align}
        \max ~&\lambda \notag\\
        \textup{s.t.}~&\lambda \geq 0, ~P \in \mathbb{S}^{n_x}_{\succ 0}, ~\tilde{Z} \in \mathbb{S}^{n_y}_{\succ 0}, ~G \in \mathbb{R}^{n_x \times n_y}, ~Y \in \mathbb{R}^{n_y \times n_a}, ~\eqref{eq: LMIs_KLD1 c1}, \notag\\
        &\begin{bmatrix}
       \Theta_1 
    &\begin{bmatrix}
        \Theta_2 &\Theta_3
    \end{bmatrix} \\
    * &\begin{bmatrix}
        1/\gamma P &0\\
        0 &\gamma P
    \end{bmatrix}
    \end{bmatrix} \succeq 0, \label{eq: LMIs_KLDinf}
    \end{align}  
    where $L=P^{-1}G$ and the involved matrices are defined as 
    \begin{align*}
        &\Theta_1 = \begin{bmatrix}
        D^{\top}_a Y + Y^{\top} D_a-\lambda D^{\top}_a W D_a &Y^{\top}\\
        * &\frac{1}{2} \tilde{Z}
    \end{bmatrix},
    ~\Theta_2 = \begin{bmatrix}
        -Y^{\top} M \\ 0
    \end{bmatrix}, 
    ~\Theta_3 = \begin{bmatrix}
        0 \\ G^{\top}
    \end{bmatrix}.
    \end{align*}
\end{theorem}
\begin{proof}
   The proof is relegated to Section~\ref{sec: proof}.
\end{proof}

We introduce the following remarks on the conservatism of approximations and the trade-off between transient and steady-state detectability.

\begin{remark}[Conservatism of the bi-convex and LMI approximations]
The bi-convex and LMI formulations constitute two different sufficient approximations of the original observer design problems.
\begin{enumerate}
    \item For the bi-convex approximations in Theorems~\ref{thm: KLD1_BMI} and~\ref{thm: KLDinf BMI}, Lemma~\ref{lem: Alt_InvMat} does not restrict the feasible set of stabilizing observer gains. 
    However, the covariance relaxation $\Sigma_{r_\omega}^{-1}(L) \succ \tilde{Z}$ may lead to a lower objective value $\lambda(L)$ than the exact value ($J_1(L)$ in the one-step case and $J_\infty(L)$ in the steady-state case).
    For any fixed stabilizing $L$, this gap can be made arbitrarily small because (1) $P$ and $\tilde Z$ can be chosen such that $\tilde Z$ approaches $\Sigma_{r_\omega}^{-1}(L)$ from below, and (2) the Young's inequality used in the bi-convex formulations can attain equality through an appropriate choice of the auxiliary variable $Y$.
    \item For the LMI formulations in Theorems~\ref{thm: KLD1LMI} and~\ref{thm: KLDinf_LMI}, additional conservatism may arise from the generalized Young's inequalities, whose tightness depends on the scalar parameter $\gamma$, for which a closed-form optimal choice is generally unavailable. 
    However, $\gamma$ can be selected through a one-dimensional grid search to reduce conservatism and improve the performance of the LMI formulations. 
\end{enumerate}
\end{remark}

\begin{remark}[Trade-off between transient and steady-state detectability]
    Sections~\ref{subsec: KLD1} and~\ref{subsec: KLDinf} investigate the detectability of worst-case BIAs in the one-step and steady-state cases, respectively. 
    Building on these results, the two objectives can be jointly considered by formulating a multi-objective optimization problem, which provides a flexible way to balance transient and steady-state detectability. This can be achieved by solving~$\max_{L \in \R^{n_x \times n_y}} \{\alpha J_1(L)+ (1-\alpha)J_{\infty}(L): \rho(A-LC)<1 \}$, where $\alpha \in [0,1]$ is the weighting parameter.
\end{remark}

\subsection{Extension to gradually injected BIAs}
In practice, an attacker can gradually inject a bias to reduce its transient detectability. 
To account for such scenarios, we next show that the above formulations can be extended
to slowly injected BIAs of the form
\begin{align}\label{eq: slow BIA}
    a(k+1)=(1-\beta)a(k)+\beta\bar a, \quad a(k_0-1)=0,
\end{align}
where $\beta\in(0,1]$ determines the injection rate and $\bar{a}$ is the target bias, e.g., the worst-case BIA obtained by solving~\eqref{eq: Opt_a}. 

\begin{lemma}[Equivalence with constant BIAs at attack onset and steady state]
    For the slowly injected BIA in~\eqref{eq: slow BIA}, the attack-onset and steady-state observer-gain design problems coincide with those for constant BIAs, except for a factor $\beta^2$ in the attack-onset objective. Hence, the corresponding observer-gain solutions remain unchanged.
\end{lemma}

\begin{proof}
    At the attack onset, $\bar{r}_a(0) = \beta D_a\bar{a}$, which scales the optimal value of~\eqref{eq: OptPro_KLD0} by $\beta^2$ but leaves the optimal gain $L$ unchanged. 
    For the steady-state case, since $a(k)\to\bar a$ as $k\to\infty$, the observer-gain design formulation coincides with that for constant BIAs. This completes the proof.
\end{proof}

In contrast, the \mbox{one-step} formulation depends explicitly on the transient attack evolution.
The attack-induced residual~$\bar{r}_a$ at $k=1$ becomes
$\bar r_a(1) = \beta\big((2-\beta)I-CL\big)D_a\bar a$, and is therefore modified in the following proposition.

\begin{proposition}[Extension of one-step detectability design to slowly injected BIAs]\label{prop: slowlyBIAdesign} 
The conditions in Theorems~\ref{thm: KLD1_BMI} and~\ref{thm: KLD1LMI} can be modified to accommodate the slowly injected BIA in~\eqref{eq: slow BIA} as follows. 
\begin{enumerate}
    \item For the bi-convex formulation in Theorem~\ref{thm: KLD1_BMI}, replace $\varphi(Y,L)$ with
    $\varphi_\beta(Y,L)
    =
    \beta(2-\beta)
    \left(Y^\top D_a+D_a^\top Y\right)
    -
    \beta\left(
    Y^\top CLD_a+
    D_a^\top L^\top C^\top Y
    \right)$.
    \item For the LMI formulation in Theorem~\ref{thm: KLD1LMI}, replace $Y^\top D_a+D_a^\top Y$ with $\beta(2-\beta)
    \left(Y^\top D_a+D_a^\top Y\right)$,
    and replace $(2+1/\gamma)^{-1}P$ in $\Gamma_2$ with
    $\left(2+\beta^2 / \gamma \right)^{-1}P$.
\end{enumerate}
For $\beta=1$, these formulations reduce to the constant-BIA case.
\end{proposition}
\begin{proof}
The result follows by replacing the one-step attack transition
$(I-CL)D_a$ in Theorems~\ref{thm: KLD1_BMI} and~\ref{thm: KLD1LMI} with
$\beta((2-\beta)I-CL)D_a$. 
The remaining derivations follow analogously and are therefore omitted.
\end{proof}
Note that the extension assumes a prescribed $\beta$, which can be treated as a defender-selected design parameter. 
Its value can be selected according to the expected range of attack injection rates, i.e., $[\underline{\beta}, ~\bar{\beta}]$.
For example, a smaller $\beta$ may be selected when slowly injected attacks are of primary concern.

%
\section{Algorithms design and residual evaluation}\label{sec: algorithms}
This section presents the two algorithms, namely, AO and ADMM, to solve the bi-convex optimization problems~\eqref{eq: BMIs_KLD1} and~\eqref{eq: BMIs_KLD_inf}, followed by the residual evaluation method. 

\subsection{AO and ADMM algorithms}\label{subsec: AO and ADMM}
Let us introduce AO first, which addresses the bi-linear terms (e.g., $PL$ and~$Y^{\top}CL$ in~\eqref{eq: BMIs_KLD1}, or $PL$ and $Y^{\top}ML$ in~\eqref{eq: BMIs_KLD_inf}) by partitioning the coupled decision variables into two disjoint sets, e.g.,~$\{P,Y\}$ and~$\{L\}$.
At each iteration of AO, one set of variables is fixed while the remaining variables are optimized, reducing the problem to a tractable convex optimization problem. 
This process alternates until convergence or some predefined stopping criteria are met.
AO is easy to implement and widely used in practice.
When all subproblems are convex, it guarantees convergence to a stationary point. However, it provides no guarantees on convergence rate or global optimality.

When using AO to solve~\eqref{eq: BMIs_KLD1} and~\eqref{eq: BMIs_KLD_inf}, the algorithm can be initialized with any gain matrix~$L$ that ensures~$\rho (A-LC)$ is Schur, e.g., the Kalman filter gain.
The algorithm terminates either when the solution converges, i.e.,~$\|L^{(i)} - L^{(i-1)}\|_2 \leq \epsilon$ where~$\epsilon \in \R_+$ is the convergence tolerance chosen by users, or when the maximum iteration number~$N_{ite}$ is reached.
The implementation details of AO in solving~\eqref{eq: BMIs_KLD1} and~\eqref{eq: BMIs_KLD_inf} are summarized in Algorithm~\ref{alg:algorithm_1}.

\begin{algorithm}[t] 
	\caption{AO for solving BIA detection observer design in~\eqref{eq: BMIs_KLD1} and~\eqref{eq: BMIs_KLD_inf}.} \label{alg:algorithm_1} 
	\begin{algorithmic}
	\State 1. \textbf{Initialize}: 
        Consider the optimization problem~\eqref{eq: BMIs_KLD1} (or~\eqref{eq: BMIs_KLD_inf}), and choose~$L^{(0)}$,~$\epsilon$, $N_{ite}$ 
        \State 2. \textbf{Set}: $i=0$, $L^{(-1)} = 0$
        \State 3. \textbf{While} $\|L^{(i)} - L^{(i-1)}\|_2 > \epsilon$ $\wedge$ $i < N_{ite}$, \textbf{do}
	\begin{itemize}
	    \item[(a)] \textbf{Step 1}: Minimization over~$(P,Y,\tilde{Z},\lambda)$ with $L^{(i)}$
                \begin{itemize}
	              \item[] $(P^{(i+1)},Y^{(i+1)}) =\arg \min\limits_{P,Y,\tilde{Z},\lambda} \{-\lambda: \eqref{eq: BMIs_KLD1} ~(\textup{or}~\eqref{eq: BMIs_KLD_inf})\}$
                \end{itemize}

            \item[(b)] \textbf{Step 2}: Minimization over~$(L,\tilde{Z},\lambda)$ with $(P^{(i+1)},Y^{(i+1)})$
	          \begin{itemize}
	             \item[] $L^{(i+1)} =\arg \min\limits_{L,\tilde{Z},\lambda} \{-\lambda: \eqref{eq: BMIs_KLD1} ~(\textup{or}~\eqref{eq: BMIs_KLD_inf})\}$
                 \end{itemize}
            \item[(c)] $i=i+1$
            \item[(d)] End while and return~$L^{(i)}$
	\end{itemize}
	\end{algorithmic} 
\end{algorithm}

In what follows, we present the approach to solving the optimization problems~\eqref{eq: BMIs_KLD1} and~\eqref{eq: BMIs_KLD_inf} using ADMM.
Let us briefly explain the mechanism of ADMM.
By introducing auxiliary variables and constraints, ADMM constructs an augmented Lagrangian and then alternately optimizes the subproblems and Lagrange multipliers to find solutions efficiently.
When applying ADMM to~\eqref{eq: BMIs_KLD1} and~\eqref{eq: BMIs_KLD_inf}, we define auxiliary variables~$G = PL$ and $Q=P^{-1}$, yielding~$[PQ ~QG] = [I ~L]$. 
The inequalities~\eqref{eq: LMIs} in ~\eqref{eq: BMIs_KLD1} and~\eqref{eq: BMIs_KLD_inf} are replaced with~\eqref{eq: LMIs_KLD1 c1} accordingly.
The following augmented Lagrangian can be constructed
\begin{align*}
    \mathcal{L}_{\eta}(\lambda,P,Q,G,L,\Lambda) 
    = -\lambda + \textup{tr}(\Lambda^{\top} ([PQ ~QG] - [I ~L])) + \frac{\eta}{2} \| [PQ ~QG] - [I ~L] \|^2_F,
\end{align*}
where $\Lambda \in \R^{n_x \times (n_x+n_y)}$ is the Lagrange multiplier associated with the constraint~$[PQ ~QG] = [I ~L]$, $\eta \in \R_+ $ is the penalty parameter.
Then, we can solve the observer gain by minimizing~$\mathcal{L}_{\eta}(\lambda,P,Q,G,L,\Lambda)$ subject to constraints in~\eqref{eq: BMIs_KLD1} (or~\eqref{eq: BMIs_KLD_inf}).
Note that $PQ$, $QG$ in~$\mathcal{L}_{\eta}(\lambda,P,Q,G,L,\Lambda)$ and $Y^{\top}CL$ in~\eqref{eq: BMIs_KLD1} (or $Y^{\top}ML$ in~\eqref{eq: BMIs_KLD_inf}) are bi-linear terms.
Thus, these coupled variables can be partitioned into two sets, i.e.,~$\{P,G,Y\}$ and $\{Q, L\}$, and then updated in an alternative fashion along with the Lagrange multiplier~$\Lambda$. 

Similar to AO, ADMM can be initiated with any $L$ that renders~$A-LC$ Schur stable. 
The $Q$ matrix is initialized using the steady-state estimation-error covariance induced by the chosen initial~$L$.
The algorithm stops when the observer gain matrix converges, i.e.,~$\|L^{(i)} - L^{(i-1)}\|_2 \leq \epsilon_1$ and $\| [PQ ~QG] - [I ~L] \|_2 \leq \epsilon_2$ or when the maximum number of iterations~$N_{ite}$ is reached.
We summarize the implementation steps of ADMM in solving~\eqref{eq: BMIs_KLD1} and~\eqref{eq: BMIs_KLD_inf} in Algorithm~\ref{alg:algorithm_2}.

\begin{remark}[Initialization with the LMI-based solutions]
The LMI-based approximations~\eqref{eq: LMIs_KLD1} and~\eqref{eq: LMIs_KLDinf} yield convex optimization problems that can be solved efficiently. Despite their conservatism due to the relaxations involved, the LMI-based solutions provide effective initializations for AO and ADMM algorithms.
\end{remark}

\begin{remark}[Complexity of AO and ADMM] 
Using the standard worst-case complexity analysis for SDP interior-point methods~\cite{jiang2020faster,nesterov1994interior}, the computational complexities of both AO and ADMM scale as
$\mathcal{O}^*\!\left(N_{\mathrm{ite}}(n_x+n_y)^{6.5}\right)$,
where $N_{\mathrm{ite}}$ is the number of outer iterations. 
Although this complexity is relatively high, both algorithms are used only for offline observer-gain synthesis. 
If online updating of $L$ is required, the computationally less demanding LMI-based method can be employed.
\end{remark}

\begin{algorithm}[t] 
	\caption{ADMM for solving BIA detection observer design in~\eqref{eq: BMIs_KLD1} and~\eqref{eq: BMIs_KLD_inf}.} \label{alg:algorithm_2} 
	\begin{algorithmic}
	\State 1. \textbf{Initialize}: 
        Consider the optimization problem~\eqref{eq: BMIs_KLD1} (or~\eqref{eq: BMIs_KLD_inf}) and the Lagrangian~$\mathcal{L}_{\eta}$, and choose~$L^{(0)}$,~$Q^{(0)}$,~$\epsilon_1$,~$\epsilon_2$,~$N_{ite}$
        \State 2. \textbf{Set}: $i=0$, $L^{(-1)} = 0$, $P^{(0)} = (Q^{(0)})^{-1}$, $G^{(0)} = P^{(0)} L^{(0)}$, $\Lambda^{(0)}=0$, $\eta = 10000$
        \State 3. \textbf{While} ($\|L^{(i)} - L^{(i-1)}\|_2 > \epsilon_1$ $\vee$ $\| [P^{(i)}Q^{ (i)} ~Q^{(i)}G^{(i)}] - [I ~L^{(i)}] \|_F > \epsilon_2$) $\wedge$ $i < N_{ite}$, \textbf{do}
	\begin{itemize}
	    \item[(a)]  \textbf{Step 1}: Minimization over~$(P,G, Y,\tilde{Z},\lambda)$ with $(L^{(i)},Q^{(i)})$
                \begin{itemize}
	              \item[] $(P^{(i+1)},G^{(i+1)},Y^{(i+1)},\tilde{Z}^{(i+1)})$$=\arg \min\limits_{P,G,Y,\tilde{Z},\lambda} \{\mathcal{L}_{\eta}: ~\eqref{eq: BMIs_KLD1} ~(\textup{or}~\eqref{eq: BMIs_KLD_inf})\}$
                \end{itemize}

            \item[(b)]  \textbf{Step 2}: Minimization over~$(L,Q,\lambda)$ with $(P^{(i+1)},G^{(i+1)},Y^{(i+1)},\tilde{Z}^{(i+1)})$
	          \begin{itemize}
	             \item[] ($L^{(i+1)},Q^{(i+1)}) =\arg \min\limits_{L,Q,\lambda} \{\mathcal{L}_{\eta}: \eqref{eq: BMIs_KLD1} ~(\textup{or}~\eqref{eq: BMIs_KLD_inf})\}$
                 \end{itemize}
                 
            \item[(c)]  \textbf{Step 3}: Update $\Lambda$
	          \begin{itemize}
	             \item[] $\Lambda^{(i+1)} =\Lambda^{(i)} + \eta ([P^{(i+1)}Q^{(i+1)} ~Q^{(i+1)}G^{(i+1)}]-[I ~L^{(i+1)}])$
                 \end{itemize}
            \item[(d)] $i=i+1$
            \item[(e)] End while and return~$L^{(i)}$
	\end{itemize}
	\end{algorithmic} 
\end{algorithm}

\subsection{Residual evaluation}
The above results provide approaches for designing detection observers such that the generated residuals~$r$ are sensitive to worst-case BIAs. 
However, since these residuals are multi-dimensional and affected by noise, they are not suitable for directly determining the occurrence of attacks.
Therefore, we introduce the squared Mahalanobis distance as the residual evaluation function, which is defined as
\begin{align*}
   \mathcal{I}(k) = r^{\top}(k) \Sigma^{-1}_{r_{\omega}} r(k).
\end{align*}
Note that the Mahalanobis distance characterizes the deviation of a sample from its underlying distribution and is widely used as a measure of statistical discrepancy.

Furthermore, as discussed in Section 2.2.2, residuals in the steady state follow a zero-mean Gaussian distribution in the absence of BIAs. 
Consequently, the evaluation function $\mathcal{I}(k)$ is subject to a chi-square distribution with $n_y$ degrees of freedom.
Its values are concentrated near the origin on the positive real axis with some probability guarantees. 
In contrast, the presence of attacks changes the residual distribution, leading to a shift in the values of $\mathcal{I}(k)$. 
Based on this observation, the following detection logic is adopted:
\begin{align*}
    \left\{ \begin{array}{ll}
        \mathcal{I}(k) \leq \tau, & \textup{normal}, \\
         \mathcal{I}(k) > \tau,    & \textup{attack detected}.
    \end{array} \right. 
\end{align*}
where the threshold $\tau \in \R_+$ is selected according to a desired false alarm rate (FAR) $\xi$, such that~$\mathbb{P}(\mathcal{I}(k) > \tau | a(k)=0) < \xi$ is satisfied.

\section{Numerical examples} \label{sec:simulation}
In this section, the effectiveness of the proposed BIA detection methods is validated through an application to a thermal system.

\begin{figure}[t]
    \centering
    \includegraphics[width=0.4\linewidth]{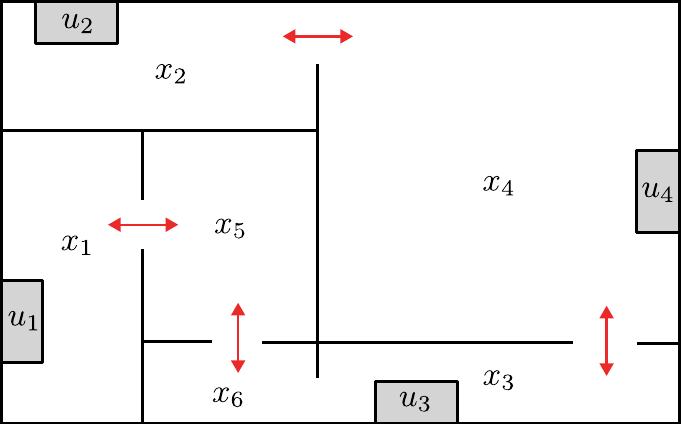}
    \caption{\small Illustration of a six-room thermal system~\cite{yang2023lasso}.}
    \label{fig: TC System}
\end{figure}

\subsection{Thermal system description}
Consider the thermal system depicted in Fig.~\ref{fig: TC System}, which consists of six rooms, and four of them are equipped with air-conditioning units.
The bidirectional arrows indicate thermal transmission through vents between adjacent rooms, with a thermal transmission rate of 0.1. 
Each vent is equipped with a sensor that measures the average temperature of the two connected rooms. 
Let~$x_{i}$, $i \in \{1,2,\dots,6\}$, denote the difference between the temperature in the $i$th room and the desired setpoint. 
The dynamics of the thermal system can be described by the following state-space model~\cite{yang2023lasso}:
\begin{align*}
    x(k+1) &= \begin{bmatrix}
        0.8 &0 &0 &0 &0.1 &0\\
        0 &0.8 &0 &0.1 &0 &0\\
        0 &0 &0.7 &0.1 &0 &0.1\\
        0 &0.1 &0.1 &0.7 &0 &0\\
        0.1 &0 &0 &0 &0.7 &0.1\\
        0 &0 &0.1 &0 &0.1 &0.7
    \end{bmatrix} x(k) + 
    \begin{bmatrix}
        I \\ 0_{2 \times 4}
    \end{bmatrix} u(k) +0.1\begin{bmatrix}
        I &0_{6 \times 5}
    \end{bmatrix} \omega(k), \\
    y(k) &= \begin{bmatrix}
        0.5 &0 &0 &0 &0.5 &0\\
        0 &0.5 &0 &0.5 &0 &0\\
        0 &0 &0.5 &0.5 &0 &0\\
        0 &0 &0 &0 &0.5 &0.5\\
        0 &0 &0.5 &0 &0 &0.5
    \end{bmatrix} x(k) + 0.1\begin{bmatrix}
        0_{5 \times 6} &I
    \end{bmatrix} \omega(k). 
\end{align*}
The system is controlled by~$u(k) = K \hat{x}_c(k)$ with a controller-side observer.  
To characterize the impact of noise, we choose matrices $B_{\omega} = 0.1 [I_6 ~0_{6 \times 5}]$ and $D_{\omega} = 0.1 [0_{5 \times 6} ~I_5]$.

Suppose that sensors $1$, $3$, and $5$ in the thermal system are compromised. 
According to the definition of the attack matrix $D_a$ in Subsection~\ref{Subsec: BIA description}, the entries $(1,1)$, $(3,2)$, and $(5,3)$ of $D_a$ are equal to $1$, and the remaining entries are zeros.
Suppose that the attack aims to increase the steady-state estimation error of the Kalman filter.
According to~\eqref{eq: error_dynamics}, we have
$\lim_{k \rightarrow \infty} \mathbb{E}[\tilde{x}(k)] = G_{y_a \tilde{x}}  D_a \bar{a}$,
where $G_{y_a \tilde{x}} = -(I - (A-L_{\rm kal}C))^{-1}L_{\rm kal}$ is the steady-state transition matrix from~$y_a$ to~$\tilde{x}$ and the weighting matrix is defined as~$W = G^{\top}_{y_a \tilde{x}} G_{y_a \tilde{x}}$.

\subsection{Stealthiness of optimal BIAs}
The stealthiness of the optimal attack is demonstrated in the following simulation results.
Given Kalman filter gain~$L_{\rm kal}$ and the weighting matrix~$W$, the optimal attack vector~$\bar{a}^{*}$ is the solution to the optimization problem~\eqref{eq: Opt_a} at $k=0$, which corresponds to the smallest generalized eigenvalue of the matrix pair ~$(\frac{1}{2}  D^{\top}_a  \Sigma_{r_{\omega}}^{-1}(L_{kal})  D_a , D^{\top}_a W D_a)$.
By calculation, we have
$\bar{a}^*_{\rm kal} = [-0.87 ~0.33 ~0.44]^{\top}$.
For comparison, we consider a random attack~$\tilde{a}$, which satisfies the same impact constraint as~$\bar{a}^*_{\rm kal}$, i.e.,~$\tilde{a}^{\top} D^{\top}_a W D_a \tilde{a} = 1$. 
The random attack used here is~$\tilde{a} = [0.8072 ~0.0307 ~0.7606]^{\top}$.
The number of simulation steps is~$300$ and the attacks are injected into the system as step signals at~$k=100$. 
The detection threshold is set to~$\tau = 16.7496$ given an acceptable false alarm rate~$\xi = 0.005$.
Fig.~\ref{fig: residual_stepattack} illustrates the Mahalanobis distance of residuals~$\mathcal{I}(k)$ generated by the Kalman filter under both the optimal attack~$\bar{a}^*_{kal}$ and the random attack~$\tilde{a}$, while Fig.~\ref{fig: impact_stepattack} shows the corresponding impacts of the two attacks on the state estimation error. 

As shown in Fig.~\ref{fig: residual_stepattack}, under the optimal attack~$\bar{a}^*_{\rm kal}$, the Mahalanobis distance of the residual exceeds the threshold after attack onset, but quickly returns and remains below the threshold. 
This indicates that~$\bar{a}^*_{\rm kal}$ exhibits strong stealthiness and has a low probability of being detected by the Kalman filter. 
In contrast, the residual induced by the random attack~$\tilde{a}$ stays above the threshold after the attack occurs, making it easier to be detected.
Fig.~\ref{fig: impact_stepattack} further shows that the optimal and random attacks yield comparable magnitudes of state estimation error. This demonstrates that the optimal attack can achieve a similar disruptive effect on the estimation performance of the Kalman filter while maintaining stealthiness.

\begin{figure}[t]
    \begin{minipage}[]{0.49\textwidth}
    \centering
    \captionsetup{justification=centering}
    \includegraphics[scale=0.6]{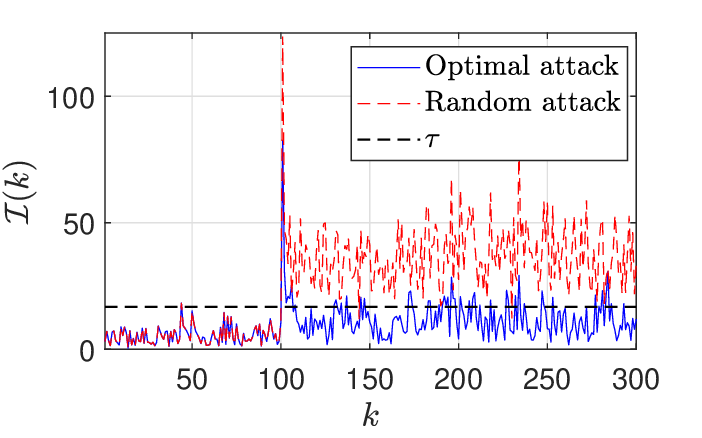} 
    \caption{\small Mahalanobis distance using Kalman filter under optimal and random step attacks.}\label{fig: residual_stepattack}
    \end{minipage} 
    \begin{minipage}{0.49\textwidth}
    \centering
    \captionsetup{justification=centering}
    \includegraphics[scale=0.6]{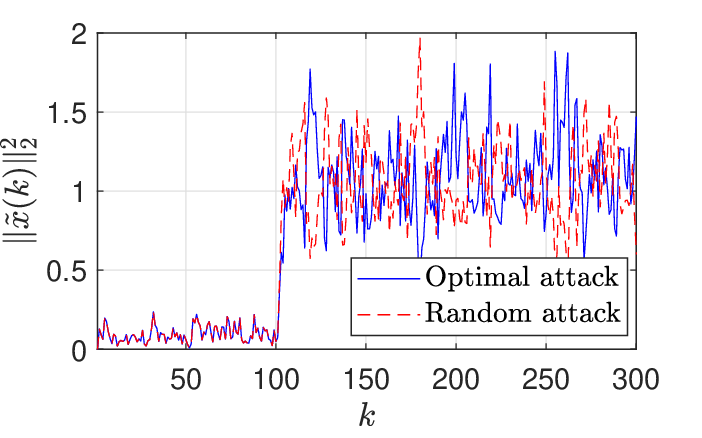} 
    \caption{\small Impacts of optimal and random step attacks on state estimation errors.}
    \label{fig: impact_stepattack}
    \end{minipage}
\end{figure}

\begin{figure}[t]
    \begin{minipage}[]{0.49\textwidth}
    \centering
    \captionsetup{justification=centering}
    \includegraphics[scale=0.6]{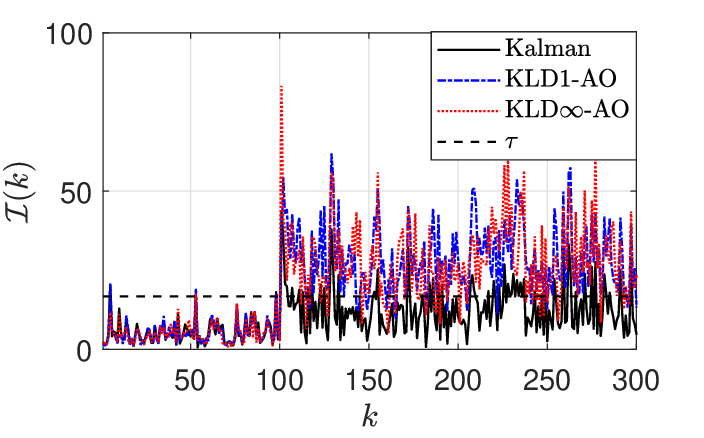} 
    \caption{\small Mahalanobis distance using KLD1 and KLD$\infty$ detectors under respective optimal attacks.}\label{fig: kldresidual_stepattack}
    \end{minipage} 
    \begin{minipage}{0.49\textwidth}
    \centering
    \captionsetup{justification=centering}
    \includegraphics[scale=0.6]{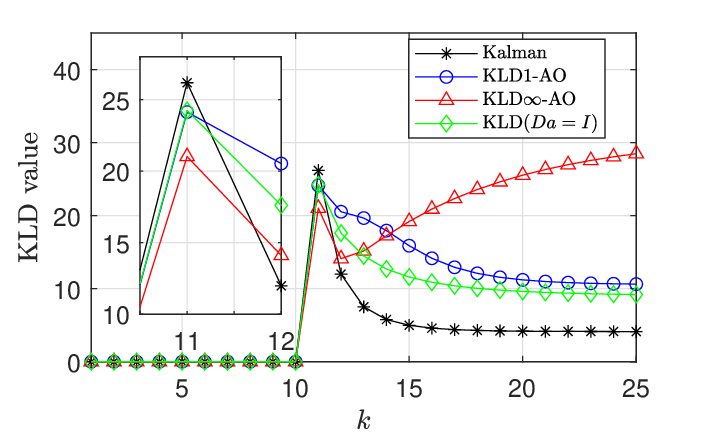} 
    \caption{\small KLD values of KLD$0$, KLD$1$, and KLD$\infty$ detectors under~$\bar{a}^*_{kal}$.}
    \label{fig: kld_stepattack}
    \end{minipage}
\end{figure}

\begin{figure}[t]
    \begin{minipage}[]{0.49\textwidth}
    \centering
    \captionsetup{justification=centering}
    \includegraphics[scale=0.6]{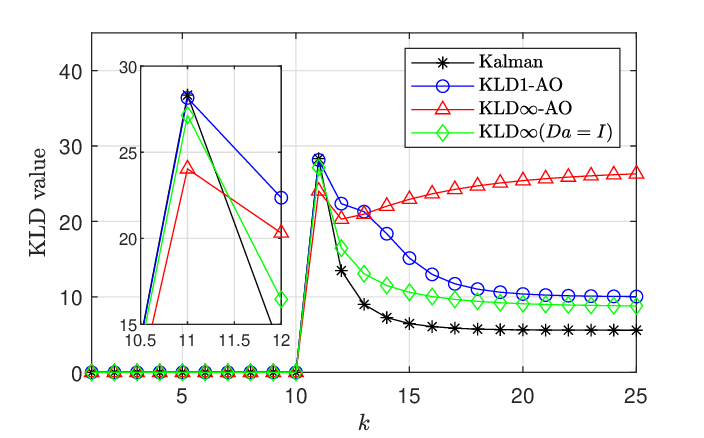} 
    \caption{\small KLD values of KLD$0$, KLD$1$, and KLD$\infty$ detectors under~$\bar{a}^*_{kld1}$ .}\label{fig: kld1value_stepattack}
    \end{minipage} 
    \begin{minipage}{0.49\textwidth}
    \centering
    \captionsetup{justification=centering}
    \includegraphics[scale=0.6]{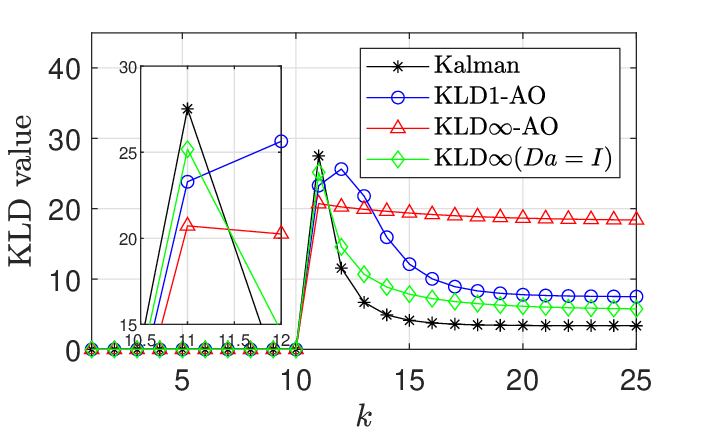} 
    \caption{\small KLD values of KLD$0$, KLD$1$, and KLD$\infty$ detectors under~$\bar{a}^*_{kld\infty}$.}
    \label{fig: kldinfvalue_stepattack}
    \end{minipage}
\end{figure}

\subsection{Step attack scenario} 
In what follows, we evaluate the performance of the proposed detectors under their corresponding worst-case BIAs.
To this end, we first compute the KLD-based observers for $k=1$ and $k\to\infty$ using the LMI methods in Theorems~\ref{thm: KLD1LMI} and~\ref{thm: KLDinf_LMI}, respectively.
Then, following the bi-convex observer design methods in Theorems~\ref{thm: KLD1_BMI} and~\ref{thm: KLDinf BMI} and setting the derived LMI-based detector gains as initial points, we further optimize the observer gains using both the AO and ADMM algorithms.
During the optimization, the maximum number of iterations is set to~$N_{\rm ite} = 50$, and the stopping tolerance is~$\epsilon = 1 \times 10^{-5}$. 
For the ADMM algorithm, the initial values for the Lagrange multiplier and the penalty parameter are chosen as~$\Lambda = 0$ and $\eta = 10000$, respectively. The better results between AO and ADMM are selected as the final observer gain, denoted by~$L^*_{\rm kld1}$ and $L^*_{\rm KLD \infty}$, respectively. 
The optimal attacks for KLD$1$ and KLD$\infty$ detectors denoted by~$\bar{a}^*_{kld1}$ and~$\bar{a}^*_{\rm kld \infty}$ are derived as 
$\bar{a}^*_{\rm kld1} = [0.9167 ~0.2650 ~0.2094]^{\top}$ and
$\bar{a}^*_{\rm kld \infty} = [-0.0924 ~0.9781 ~0.8206]^{\top}$, respectively.
To illustrate the advantages of explicitly considering attacks on subsets of sensors, we also consider the case where $D_a = I$ and design the detection observer to enhance the attack detectability in the steady-state.

Fig.~\ref{fig: kldresidual_stepattack} shows the Mahalanobis distances of residuals generated by KLD$1$ and KLD$\infty$ detectors under their respective worst-case BIAs. 
As illustrated, the Mahalanobis distances for both detectors immediately exceed the threshold upon the attack onset and remain above it thereafter.
This demonstrates that the proposed detectors are able to detect the worst-case attacks, while the traditional Kalman filter fails to do so.

Figs.~\ref{fig: kld_stepattack}-\ref{fig: kldinfvalue_stepattack} depict the KLD values by the proposed KLD-based detectors under the three optimal attacks: $\bar{a}^*_{\rm kal}$,~$\bar{a}^*_{\rm kld1}$ and $\bar{a}^*_{\rm kld\infty}$.
A consistent pattern can be observed across all figures.
Specifically, as the optimal solution for the case~$k=0$, the Kalman filter achieves the largest KLD value at attack onset.
In contrast, the KLD$1$ detector corresponding to the optimal solution for $k=1$ attains the largest KLD value at the next instant after the attack, demonstrating the best detection capability during the transient phase.
Finally, the KLD$\infty$ detector yields the largest KLD value as the system approaches steady state.
However, when $D_a=I$ is used in the design, the resulting KLD$\infty$ detector achieves smaller KLD values than the KLD$1$ and KLD$\infty$ detectors designed with the structured $D_a$.
This highlights the benefit of explicitly accounting for partially compromised sensors in the detector design.

\begin{figure}[t]
    \begin{minipage}[]{0.49\textwidth}
    \centering
    \captionsetup{justification=centering}
    \includegraphics[scale=0.6]{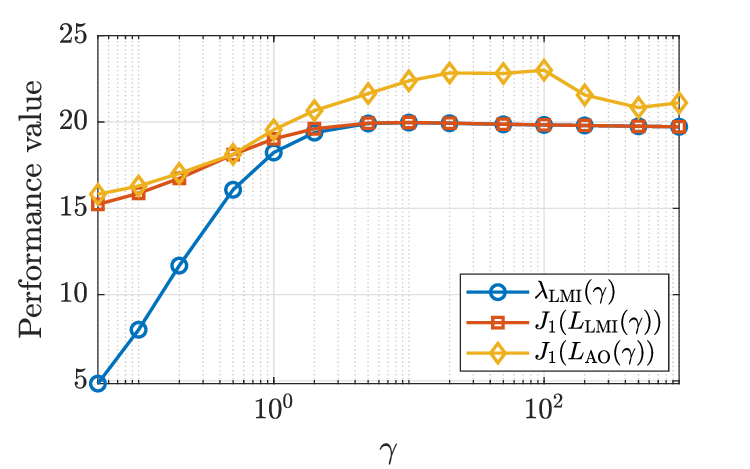} 
    \caption{\small Effect of $\gamma$ on the KLD1 LMI relaxation.}\label{fig: kld1_gamma}
    \end{minipage} 
    \begin{minipage}{0.49\textwidth}
    \centering
    \captionsetup{justification=centering}
    \includegraphics[scale=0.6]{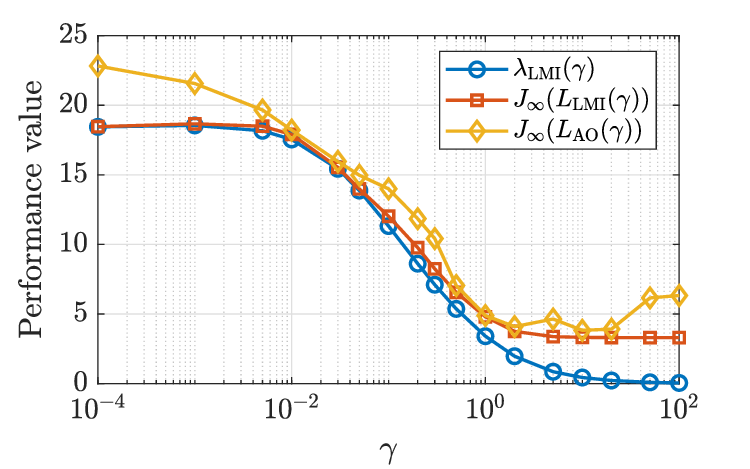} 
    \caption{\small Effect of $\gamma$ on the KLD$\infty$ LMI relaxation.}
    \label{fig: kldinf_gamma}
    \end{minipage}
\end{figure}

Figs.~\ref{fig: kld1_gamma} and~\ref{fig: kldinf_gamma}
illustrate the influence of $\gamma$ on the LMI relaxations through a grid search. 
Here, $\lambda_{\rm LMI}(\gamma)$ is the objective value of the LMI relaxation, $J_k(L_{\rm LMI}(\gamma))$ is the actual KLD value achieved by the resulting LMI-based gain, and $J_k(L_{\rm AO}(\gamma))$ is the KLD value after further refinement using AO, where $k=1$ for Fig.~\ref{fig: kld1_gamma} and $k=\infty$ for Fig.~\ref{fig: kldinf_gamma}.
The KLD1 relaxation becomes tighter as $\gamma$ increases, while the KLD$\infty$ relaxation favors relatively small values of
$\gamma$. 
The figures also show that the detector performance can be further improved by using the LMI-based observer gains as initializations for the AO approach.


\subsection{Stealthy attack scenario.}
In this part, we consider the slowly injected BIA in~\eqref{eq: slow BIA}. 
The actual injection rate is~$\beta = 0.01$, and the target bias is chosen as~$\bar{a} = \bar{a}^*_{kld \infty}$. 
Assume that sensors $1$, $2$, $3$, and $5$ in the thermal system are compromised in this part. 
The attack impact matrix~$W$ is kept the same as the previous simulation example.
Based on the above setup, the KLD$1$ and KLD$\infty$ detectors are developed using the methods in Theorems~\ref{thm: KLD1_BMI} and~\ref{thm: KLDinf BMI}, respectively, with the same design procedures and parameters as in the step-attack simulation.  
In addition, we construct a KLD$1$ detector using Proposition~\ref{prop: slowlyBIAdesign}, denoted by KLD$1$-ADMM($\beta$). 
The design parameter is set to $\beta=0.1$, which is intentionally chosen to be different from the actual attack injection rate $0.01$.

To demonstrate the advantages of the proposed detectors, we compare them with three representative methods: 
(1) A robust residual generator~\cite{dong2025robust}, which employs mixed $\mathcal{H}_2/\mathcal{H}_{-}$ indices to enhance anomaly sensitivity while attenuating the effect of noise;
(2) An LSTM-based prediction detector~\cite{zhang2022anomaly}, which learns the nominal input-output dynamics from attack-free data and detects anomalies based on the output prediction error;
(3) A hybrid Kalman/PCA detector~\cite{qin2012survey}, which generates residuals using the Kalman filter and then applies PCA-based statistical monitoring to a finite residual window.
For a fair comparison, the thresholds of all detectors are calibrated to achieve the same prescribed FAR~$\xi = 0.005$.

\begin{figure}[t]
    \begin{minipage}[]{0.49\textwidth}
    \centering
    \captionsetup{justification=centering}
   \includegraphics[scale=0.6]{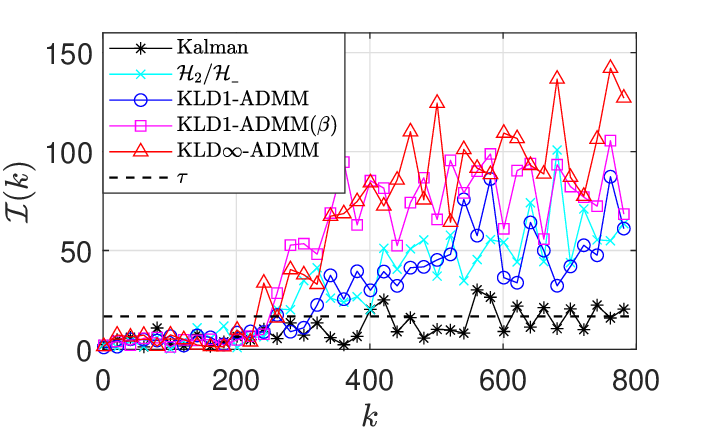} 
    \caption{\small Mahalanobis distance by different detectors under the stealthy attacks.}
    \label{fig: residual_stealthyattack}
    \end{minipage} 
    \begin{minipage}{0.49\textwidth}
    \centering
    \captionsetup{justification=centering}   
     \includegraphics[scale=0.6]{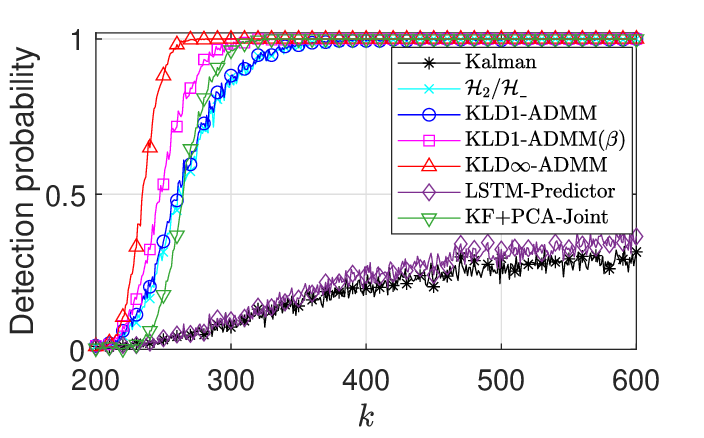} 
    \caption{\small Detection probability at each step under the stealthy attacks.}\label{fig: probdet_stealthyattack}
    \end{minipage}
\end{figure}

The simulation results are presented in Figs.~\ref{fig: residual_stealthyattack} and \ref{fig: probdet_stealthyattack}, where the simulation horizon is $800$ steps, and the attack occurs at $k_0=200$.
A data point is plotted every $20$ steps.
Particularly, Fig.~\ref{fig: residual_stealthyattack} presents the Mahalanobis distances generated by different detectors under the slowly injected BIA.
As shown in the figure, the Mahalanobis distance of the Kalman filter oscillates around the threshold after the attack occurs, preventing timely and reliable detection.
The $\mathcal{H}_2/\mathcal{H}_{-}$ filter and KLD$1$ detector exhibit comparable performance.
The KLD$1$-ADMM($\beta$) detector, which explicitly accounts for the transient attack evolution, improves the detection performance of the KLD$1$ design, even under a mismatch between the design and actual injection rates.  
In contrast, the KLD$\infty$ detector yields the largest Mahalanobis distance among all considered methods, demonstrating strong steady-state detection capability.

Fig.~\ref{fig: probdet_stealthyattack} shows detection probabilities of different detectors over time via Monte Carlo simulations. 
The detection probabilities of the KLD$1$, KLD$\infty$, $\mathcal{H}_2/\mathcal{H}_{-}$, KLD$1$-ADMM($\beta$), and KF+PCA-joint detectors eventually approach one, whereas those of the Kalman filter and the LSTM-based detector increase much more slowly and remain significantly below one.
Among all methods, the KLD$\infty$ detector exhibits the fastest increase in detection probability and achieves the best overall detection performance, followed by the KLD$1$-ADMM($\beta$) detector.
The $\mathcal{H}_2/\mathcal{H}_{-}$ and KF+PCA-joint detectors show comparable performance. 
In particular, the KF+PCA-joint detector achieves a high detection probability at later stages because it exploits a finite window of residuals. 
However, this window-based accumulation also introduces additional detection delay compared with the proposed single-instant approaches.

\begin{table}[t]
    \centering
    \caption{\small Empirical FAR and ADD under different levels of model uncertainty. Entries are reported as FAR(\%)/ADD (samples).}\label{tab:FAR_ADD} 
    \begin{tabular}{c c c c}
    \hline
    $\Delta$ &\text{Kalman} & \text{KLD1}-ADMM & $\text{KLD}_{\infty}$-ADMM \\
    \hline
     0\% & 0.50/43.7 & 0.50/29.2 & 0.50/25.0 \\
     5\% & 0.50/44.1 & 0.51/29.2 & 0.54/24.9 \\
    10\% & 0.54/43.3  & 0.58/28.8  & 0.67/24.3 \\
    15\% & 0.60/42.2 & 0.85/28.1 & 1.32/23.8 \\
    20\% & 4.02/39.0 & 6.00/25.4 & 8.00/20.9 \\
    \hline
    \end{tabular}    
\end{table}

Furthermore, we investigate the impact of model uncertainty on the FAR and average detection delay (ADD) of the proposed KLD-based detectors, as summarized in Table~I. 
The nominal system matrices are perturbed at different uncertainty levels, denoted by~$\Delta$.
Up to $15\%$ perturbation, the detection curves remain close to the nominal case, although FAR increases with the mismatch. 
At an uncertainty level of $20\%$, a sharp increase in the FAR appears, indicating a significant performance degradation.
The decrease in ADD is mainly caused by the increased residual fluctuations under larger model mismatch.


\section{Conclusions} \label{sec:conslusion}
This paper proposed an observer-gain design framework for detecting BIAs targeting subsets of sensors in linear systems subject to Gaussian noise.
The observer design was formulated as a max-min problem to enhance the KLD-based detectability of worst-case BIAs subject to a prescribed attack impact. 
The Kalman filter was shown to be optimal at attack onset, while bi-convex formulations and LMI relaxations were developed for the one-step and steady-state cases. 
The resulting KLD$1$ and KLD$\infty$ observers both improve the detectability of BIAs, thereby enabling earlier warning for subsequent defense actions, e.g., sensor isolation or attack compensation, to reduce the subsequent attack impact.
Future work will consider active detection mechanisms for broader classes of adversarial attacks.

\appendix
\renewcommand{\thesection}{\Alph{section}}
\section{Technical Proofs}\label{sec: proof}
We first provide some useful matrix relations and inequalities that are instrumental in later proofs.

\begin{lemma}[Matrix relations]\label{lem: Matrix_Ineq}
    Let $S$ be a positive definite matrix,~$X$ and $Y$ be matrices of appropriate dimensions. The following matrix relations hold:
    \begin{enumerate}
        \item \textbf{Young's relation:} $X^\top S^{-1}X  \succeq X^\top Y + Y^\top X - Y^\top S Y$ as $(X-SY)^\top S^{-1} (X-SY) \succeq 0$. 
        \item \textbf{Generalized Young's relation:} $X^{\top} S^{-1} Y+ Y^{\top} S^{-1} X \preceq \frac{1}{\gamma}X^{\top} S^{-1} X + \gamma Y^{\top} S^{-1} Y$ for $\gamma \in \R_+$, because~$(\frac{1}{\sqrt{\gamma}}X^{\top} - \sqrt{\gamma}Y^{\top})S^{-1}(\frac{1}{\sqrt{\gamma}}X - \sqrt{\gamma}Y) \succeq 0$.
        \item \textbf{Binomial inverse theorem:} $(S+XY)^{-1}=S^{-1}-S^{-1}X\left(I+YS^{-1}X\right)^{-1}YS^{-1}$ provided $I+YS^{-1}X$ is not singular. 
    \end{enumerate}
\end{lemma}

\begin{proof}[Proof of Theorem~\ref{thm: KLD1LMI}]
Applying Schur complement to~\eqref{eq: LMIs_KLD1 c2} leads to
\begin{align}\label{eq: KLD1_LMI_proof1}
        &\begin{bmatrix}
         Y^{\top}D_a + D^{\top}_a Y   &D^{\top}_a R^{\top}_W \\
         R_W D_a &2 \mu I
        \end{bmatrix}
        - \begin{bmatrix}
            \Gamma_1 &0\\
            0        &\mu I
        \end{bmatrix} \begin{bmatrix}
            \Gamma_2 &0\\
            0 &\mu I
        \end{bmatrix}^{-1} \begin{bmatrix}
            \Gamma^{\top}_1 &0\\
            0        &\mu I
        \end{bmatrix} \notag \\
        = &\begin{bmatrix}
         Y^{\top}D_a + D^{\top}_a Y   &D^{\top}_a R^{\top}_W \\
         R_W D_a &2 \mu I
        \end{bmatrix} - 
        \begin{bmatrix}
        \Gamma_1 \Gamma^{-1}_2 \Gamma^{\top}_1 &0\\
        0 &\mu I
    \end{bmatrix} \succeq 0,
\end{align}
where~$\Gamma_1 \Gamma^{-1}_2 \Gamma^{\top}_1 = (2+1/\gamma)Y^{\top}CP^{-1}C^{\top}Y + \gamma D^{\top}_a G^{\top} P^{-1} G D_a + 2Y^{\top}D_{\omega} D^{\top}_{\omega} Y$.
Based on the generalized Young's relation in Lemma~\ref{lem: Matrix_Ineq}, it holds that
\begin{align*}
 \frac{1}{\gamma} Y^{\top}CP^{-1}C^{\top}Y + \gamma D^{\top}_a G^{\top} P^{-1} G D_a   
 \succeq  Y^{\top}CP^{-1}GD_a + D^{\top}_a G^{\top} P^{-1} C^{\top}Y.
\end{align*}
Therefore, $\Gamma_1 \Gamma^{-1}_2 \Gamma^{\top}_1 \succeq 2Y^{\top}(CP^{-1}C^{\top} + D_{\omega} D^{\top}_{\omega})Y + Y^{\top}CP^{-1}GD_a + D^{\top}_a G^{\top} P^{-1} C^{\top}Y$. 
Together with the inequality~\eqref{eq: KLD1_LMI_proof1}, we have
\begin{align*}
    &\begin{bmatrix}
         Y^{\top}D_a + D^{\top}_a Y   &D^{\top}_a R^{\top}_W \\
         R_W D_a &2 \mu I
        \end{bmatrix}
    - 
    \begin{bmatrix}
        \begin{array}{l}
             Y^{\top}CP^{-1}GD_a + D^{\top}_a G^{\top} P^{-1} C^{\top}Y \\+ 2 Y^{\top}(CP^{-1}C^{\top}+D_{\omega} D^{\top}_{\omega})Y 
        \end{array}  &0\\
        0 &\mu I
    \end{bmatrix} \\
    = &\begin{bmatrix}
         \varphi(Y,L) - 2 Y^{\top} (C P^{-1} C + D_{\omega}D^{\top}_{\omega}) Y  &D^{\top}_a R^{\top}_W \\
         R_W D_a &\mu I
    \end{bmatrix}
    \succeq 0,
\end{align*}
where $L = P^{-1}G$ and $\varphi(Y,L) := Y^{\top}(I-CL)D_a + D_a^{\top}(I-CL)^{\top}Y$. 
Again, utilizing the Schur complement yields
\begin{align*}
 \varphi(Y,L) - 2 Y^{\top} (C P^{-1} C + D_{\omega}D^{\top}_{\omega}) Y - \lambda D^{\top}_a W D_a \succeq 0,
\end{align*}
where~$\mu = 1/\lambda$ and~$W = R^{\top}_W R_W$.
Moreover, according to Lemma~\ref{lem: Alt_InvMat}, inequalities~\eqref{eq: LMIs} ensure that $\rho(A-LC)<1$ and $\Sigma_{r_{\omega}}^{-1}(L) \succ (C P^{-1} C^{\top} + D_{\omega} D_{\omega}^{\top} )^{-1} \succ \tilde{Z}$.
Therefore, it holds that~$\varphi(Y,L) - 2Y^{\top} \Sigma_{r_{\omega}}(L) Y - \lambda D^{\top}_a W D_a \succeq 0$.
Recall the inequality~\eqref{eq: KLD1_BMI_proof 1}, we obtain
\begin{align*}
     &\frac{1}{2}D^{\top}_a (I-CL)^{\top}\Sigma^{-1}_{r_{\omega}}(L) (I-CL) D_a - \lambda D^{\top}_a W D_a \succeq 0,
\end{align*}
This proves the sufficiency of the conditions~\eqref{eq: LMIs_KLD1} for~\eqref{eq: OptPro_KLD1}.
Finally, by maximizing $1/\mu$ (equivalent to minimizing $\mu$) with the obtained sufficient conditions, we obtain the approximated optimization problem in Theorem~\ref{thm: KLD1LMI}, which completes the proof.
\end{proof}

\begin{proof}[Proof of Lemma~\ref{lem: Alternative trans matrix}]
Recall that $M=C(I-A)^{-1}$. 
We first show that $I+ML$ is nonsingular. By Sylvester's determinant identity,
\begin{align*}
    \det(I+ML) = \det(I+C(I-A)^{-1}L) = \det(I+(I-A)^{-1}LC),
\end{align*}
where $\det(\cdot)$ denotes the determinant of a square matrix.  
Since
$I+(I-A)^{-1}LC=(I-A)^{-1}(I-A+LC)$,
it follows that 
$\det(I+ML)=\det((I-A)^{-1}(I-A+LC))\neq0$,
because $A$ has no unit eigenvalues and $A-LC$ is Schur, implying that both $I-A$ and $I-A+LC$ are nonsingular.
Therefore, $I+ML$ is nonsingular.

Setting $S=I-A$, $X=L$, and $Y=C$ in the Binomial inverse theorem in Lemma~\ref{lem: Matrix_Ineq} yields
\begin{align*}
    &C(I-A + LC)^{-1} L \\
    =&C (I-A)^{-1}L - C(I-A)^{-1} L\left(I + C(I-A)^{-1}L\right)^{-1} C(I-A)^{-1}   L \\
    =&ML - ML(I+ML)^{-1}ML.
\end{align*}
Then, we have
\begin{align*}
    \Phi(\infty,L) &= I - ML +  ML(I+ML)^{-1}ML \\
    &= I - (I+ML)(I+ML)^{-1}ML + ML(I+ML)^{-1}ML \\
    &=  I - (I+ML)^{-1}ML = (I+ML)^{-1}. 
\end{align*}
This completes the proof.
\end{proof}

\begin{proof}[Proof of Theorem~\ref{thm: KLDinf_LMI}]
    Applying the Schur complement to~\eqref{eq: LMIs_KLDinf} leads to 
    \begin{align*}
    \Theta_1
    - \begin{bmatrix}
        \Theta_2 &\Theta_3
    \end{bmatrix}
    \begin{bmatrix}
        \gamma P^{-1} &0\\
        0 &\frac{1}{\gamma} P^{-1}
    \end{bmatrix}
    \begin{bmatrix}
        \Theta^{\top}_2 \\ \Theta^{\top}_3
    \end{bmatrix}  
    =\Theta_1 - \gamma \Theta_2P^{-1} \Theta^{\top}_2
    -\frac{1}{\gamma} \Theta_3 P^{-1} \Theta_3^{\top} 
    \succeq 0.
    \end{align*}
According to the generalized Young's relation in Lemma~\ref{lem: Matrix_Ineq}, it holds that 
\begin{align*}
     &-\Theta_2 P^{-1} \Theta_3^{\top}
    -\Theta_3 P^{-1} \Theta_2^{\top}
    \succeq 
    -\gamma \Theta_2 P^{-1} \Theta_2^{\top}
    -\frac{1}{\gamma} \Theta_3 P^{-1} \Theta_3^{\top}.
\end{align*}
As a result, we have
\begin{align*}
    &\Theta_1 -\Theta_2 P^{-1} \Theta_3^{\top} -\Theta_3 P^{-1} \Theta_2^{\top} \\
    =&\begin{bmatrix}
        D^{\top}_a Y + Y^{\top} D_a-\lambda D^{\top}_a W D_a &Y^{\top}(I+ML)\\
        * &\frac{1}{2} \tilde{Z}
    \end{bmatrix} \succeq 0,
\end{align*}
where $L=P^{-1}G$.
Again, utilizing the Schur complement, the above inequality can be further written as
\begin{align*}
    D^{\top}_a Y + Y^{\top} D_a - 2Y^{\top}(I+ML) \tilde{Z}^{-1} (I+ML)^{\top} Y - \lambda D^{\top}_a W D_a \succeq 0.
\end{align*}
Based on Lemma~\ref{lem: Alt_InvMat}, conditions~\eqref{eq: LMIs_KLD1 c1} ensure that $\rho(A-LC)<1$ and $\Sigma_{r_{\omega}}^{-1}(L) \succ \tilde{Z}$.
As a result, the following inequalities hold
\begin{align*}
    &\frac{1}{2}D^{\top}_a (I + ML )^{-\top} \Sigma_{r_{\omega}}(L)^{-1} (I + ML )^{-1} D_a - \lambda D^{\top}_a W D_a \\
    \succeq 
    &D^{\top}_a Y + Y^{\top} D_a - 2Y^{\top}(I+ML) \Sigma_{r_{\omega}}(L) (I+ML)^{\top} Y - \lambda D^{\top}_a W D_a 
    \succeq 0,
\end{align*}
where the first inequality has been proved in~\eqref{eq: KLDinf_BMI_proof 1}. 
The satisfaction of constraints in~\eqref{eq: OptPro_KLD_inf} under conditions in~\eqref{eq: LMIs_KLDinf} is thus demonstrated.
This completes the proof.
\end{proof}

\bibliographystyle{elsarticle-num}
\bibliography{ref}
\end{document}